\documentclass[11pt]{article}
\usepackage{amsmath}
\usepackage{amsfonts}
\usepackage{array}
\usepackage{graphicx}
\usepackage{xcolor}
\usepackage[backref=true]{hyperref}
\usepackage{booktabs}
\usepackage{caption}
\usepackage{amsthm}
\usepackage{accents}
\usepackage{bm}

\newtheorem{remark}{Remark}%
\usepackage{algorithm}
\usepackage[
    a4paper,
    left=1.5cm,
    right=1.5cm,
    top=1.5cm,
    bottom=1.5cm
]{geometry}
\usepackage{setspace}
\usepackage{algpseudocode}
\usepackage{subcaption}
\usepackage{natbib}
\usepackage{authblk}
\numberwithin{equation}{section}
\newtheorem{theorem}{Theorem}
\newtheorem{lemma}{Lemma}
\newtheorem{corollary}{Corollary}
\usepackage{hyperref}
 \hypersetup{
    colorlinks=true,
    linkcolor=blue,
    filecolor=blue,      
    urlcolor=blue,
    citecolor=blue
}
\newcounter{assump} 
\renewcommand{\theassump}{AN\arabic{assump}}

\newenvironment{assumption}[1][]{
  \refstepcounter{assump}%
  \par\noindent
  {\theassump.}
}{\par\normalfont}
\newcounter{condition}
\renewcommand{\thecondition}{C\arabic{condition}}

 \usepackage[T1]{fontenc}

\begin{document}
\title{\textbf{Extreme Population Selection under Multistage Sampling design With Applications}}
\date{}
\author{Shivam\textsuperscript{a} , Bhargab Chattopadhyay\textsuperscript{b},
and
  Nil Kamal Hazra\textsuperscript{a}\\
\textsuperscript{a}Department of Mathematics, \\Indian Institute of Technology Jodhpur, Jodhpur, Rajasthan, India - 342030; \\
\textsuperscript{b} School of Management and Entrepreneurship,\\ Indian Institute of Technology Jodhpur, Jodhpur, Rajasthan, India-342030}
\maketitle
\doublespacing
\noindent \textbf{Abstract:} We study the problem of selecting the extreme (best or worst) population from among $K(\geq 2)$  populations, under the assumption that the extreme population is sufficiently separated from the nearest population. The selection is based on an appropriate measure, which may vary across different application domains. Since the actual value of the measure is unknown, we obtain its estimator using the generalized method of moments under a multistage sampling design. Using this estimator, we propose two sequential procedures, namely online algorithm and multi armed bandit based algorithm. Under suitable regularity conditions and without imposing parametric assumptions on the underlying distributions, both algorithms correctly identify the extreme population with a desired level of confidence. We illustrate the proposed algorithms through applications in econometrics and genetics. In the econometric application, the extreme population is selected using the Gini index as a measure of inequality and the performance of the proposed procedures is assessed through extensive Monte Carlo simulation studies conducted under various distributional settings. In the genetics application, the worst population is identified using a measure derived from the tumor mutation burden (TMB) score and the practical applicability of the proposed algorithms is demonstrated using the Memorial Sloan Kettering-IMPACT 50000 clinical sequencing cohort. Further, we use the proposed framework to identify an anomalous population, provided such a population exists.


\section{Introduction}\label{3.S1}
Identifying or selecting the extreme (best or worst) population among several populations is an important problem in decision-making under uncertainty. As an example, consider a clinician comparing several drugs to identify the one that provides the greatest benefit to patients suffering from a disease. If this decision is based on data from only a small number of patients, the risk of incorrectly concluding that an ineffective drug is beneficial becomes much higher. Small samples also introduce substantial variability, making the results less likely to replicate and potentially leading to clinical decisions that compromise patient safety. For more details on the effects of small sample sizes, we refer to \cite{rajput2023evaluation} and \cite{faber2014sample}. On the other hand, collecting too many observations increases both time and resource requirements. Therefore, the objective is to identify the best (or worst) population with a pre-specified probability of correct selection while using as few observations as possible. Such selection problems arise in many application domains, a few of which are discussed below.
As an example, in cancer genomics, identifying the population least likely to benefit from a given therapy is important for guiding treatment decisions. Tumor mutational burden (TMB) is a widely used biomarker and a standard predictor of immunotherapy response. \cite{samstein2019tumor} showed that tumors with higher TMB respond better to immune checkpoint inhibitors across multiple cancer types. Comparing TMB scores across genetic ancestry groups can therefore help identify the population least likely to benefit from immunotherapy. \cite{luo2025ancestral} reported ancestry-related differences in treatment efficacy, while \cite{tanwar2023understanding} found that Indian patients had significantly higher TMB scores for several cancer types than the predominantly Caucasian TCGA cohort. Thus, if a particular ancestry group consistently exhibits the lowest TMB, immunotherapy may be a less effective first-line treatment for that group, prompting clinicians to consider alternative therapies. Similarly, from an economic perspective, an important problem is identifying the best region or country using indicators such as per capita income, unemployment rate, poverty level, economic growth, or income inequality. For example, policymakers may be interested in identifying the region with the lowest level of income inequality. \cite{horrace2008ranking} proposed methods for identifying populations with the lowest and highest income inequality using the Gini coefficient and related inequality measures. Similarly, the Organization for Economic Co-operation and Development (OECD)\footnote{\url{https://www.oecd.org/en/publications/society-at-a-glance-2024_918d8db3-en/full-report/income-and-wealth-inequalities_7ac4178f.html\#figure-d1e11015-0ee0a256aa}} regularly compares Gini indices across countries to identify those with the lowest and highest levels of disposable income inequality. Their findings indicate that Nordic and several Central European countries exhibit the lowest inequality levels, whereas Latin American countries, Türkiye, and the United States exhibit comparatively higher inequality. The same problem arises in several other fields. In agriculture, researchers compare crop varieties based on yield to identify the most productive variety \citep[for example,][]{gebeyaw2024participatory, agronomy12040765}. In clinical trials, \cite{https://doi.org/10.1002/pst.70023} proposed a method for selecting the most suitable personalized treatment using multiple ordinal or binary outcomes. Depending on the application, a population may represent a treatment group in a clinical trial, an ancestry group in genetics, a crop variety in agriculture, or any other comparable entity of interest. These examples illustrate that the problem of selecting an extreme population is fundamental across diverse scientific disciplines. 

The selection of an extreme population depends not only on the performance measure or parameter of interest, which varies across applications, but also on how that parameter is estimated because its true value is generally unknown. The choice of sampling design therefore plays an important role in determining the reliability of the estimator. For relatively homogeneous populations, simple random sampling is often adequate. However, many real-world populations are heterogeneous, and sampling designs that account for stratification and clustering generally produce more reliable estimates \citep[for example,][]{BD2005, BD2007}. Regardless of the sampling design, a fixed sample size may still result in an incorrect selection if too few observations are collected, whereas collecting too many observations increases data collection costs and delays decision-making. To balance selection accuracy with sampling cost, sequential sampling procedures can be employed, where observations are collected in stages and sampling continues until a pre-specified stopping criterion is satisfied \citep[for example,][]{chattopadhyay2022minimum, shivam2026asymptotic, XING2026105640}.

Under simple random sampling, several sequential algorithms have been developed for this problem. These methods broadly fall into two categories: ranking-and-selection (R\&S) procedures and pure-exploration multi-armed bandit (MAB) algorithms. Both can be used to identify the extreme population while achieving a pre-specified probability of correct selection. For further details, see \cite{chen2017nearly}, \cite{jun2016top}, \cite{karnin2013almost}, and \cite{even2006action}, among others. Among the many application areas, clinical trials have been one of the primary motivations for the development of sequential R\&S procedures. In this setting, several new treatments are compared with an existing standard treatment to determine whether any new treatment provides superior performance and should be adopted in practice. This comparison-with-a-standard problem has been studied extensively in the ranking-and-selection literature. For example, \cite{nelson2001comparisons} and \cite{rollin2005two} proposed two-stage procedures under the assumption of normality. \cite{kim2005comparison} subsequently developed a fully sequential procedure, and \cite{kim2006asymptotic} extended this work by removing the normality assumption. Specifically, they proposed two fully sequential procedures for identifying the extreme population. Additional developments in this area can be found in \cite{chick2012sequential}, \cite{xie2013sequential}, \cite{Dudewicz01021980}, and the references therein.

\subsection{Main contributions of this Article}
The existing literature assumes that the observations within each population are independent and identically distributed. Moreover, most studies consider normally distributed populations, employ the population mean as the performance criterion, and utilize simple random sampling for parameter estimation. However,  real-world populations are heterogeneous, and simple random sampling may fail to provide reliable estimates of the population characteristics required for best selection. Such estimation inaccuracies may lead to erroneous identification of the true extreme population.  Moreover, these algorithms are primarily  designed for scenarios in which the selection criterion is based on the population mean, thereby limiting their applicability to broader practical scenarios.

Motivated by the aforementioned limitations, this study makes several methodological contributions. In particular, we introduce a generalized measure and propose its estimation using the generalized method of moments (GMM) approach. The proposed generalized measure provides a unified framework that encompasses a wide range of performance measures utilized across various domains for extreme population selection, without being restricted to the population mean. In contrast to the simple random sampling design commonly assumed in the existing literature on extreme population identification, we estimate the proposed measure under a multistage sampling design that explicitly accounts for both stratification and clustering. This makes the estimator suitable for the heterogeneous populations encountered in practice. Here, we propose two sequential selection algorithms, namely a multi-stage (or online) algorithm and a multi-armed bandit (MAB)-based algorithm, for identifying the extreme population with a pre-specified confidence level. Unlike existing methods that generally rely on normality assumptions, the proposed algorithms offer a distribution-free framework. Under mild regularity conditions, we establish the theoretical properties of the proposed algorithms.
Finally, we demonstrate the practical utility of the proposed framework through an extensive simulation study and a real-data analysis. The simulation study considers the Gini index as a measure for the best selection, while the real-data analysis is conducted based on tumor mutational burden (TMB) scores
from the Memorial Sloan Kettering–IMPACT 50,000 clinical sequencing cohort, comprising more than 54,000 tumors. 
Additionally, we discuss an extension of the proposed framework for detecting anomalous populations when present.

The rest of the paper is organized as follows. Section \ref{3.S2.0} presents a multistage sampling design considered in this study. We then state the problem formulation and define a generalized measure, followed by the construction of an estimator for this measure under the multistage sampling framework. In Section \ref{3.S3}, we develop two sequential algorithms for detecting the extreme population. The theoretical results and simulation studies of the proposed methods are presented in Section \ref{3.S4}. Section \ref{3.S5} demonstrates the performance of the proposed algorithms using the MSK-IMPACT 50000 clinical sequencing cohort data. In Section \ref{3.S6}, we use the proposed methodology to identify anomalous populations, when they exist. Finally, Section \ref{3.S7} provides concluding remarks. To enhance readability of the paper, all proofs of the main results and associated supporting lemmas are provided in the Appendix.


\section{Multistage Sampling Design, Problem Formulation and Estimator of a Generalized Measure}\label{3.S2.0}
In this section, we first formulate the selection problem by specifying which population is considered the extreme among the given K independent populations. We then describe the multistage sampling design considered in the paper. We then develop an estimator for the generalized measure employed in the best population selection procedure. Throughout the paper, the $K$ populations are indexed by $i=1,2,\ldots,K$.
\subsection{Problem statement}\label{3.S2.1}
Suppose there are $K$ independent populations, indexed by $i=1,2,\ldots,K$. Let $\bm{\theta}_i=(\theta_{i1},\theta_{i2},\dots,\theta_{il})$ is the $l$ dimensional parameter vector of interest and $g(\bm{\theta}):\mathbb{R}^l \rightarrow \mathbb{R}$ is a continuous function having first-order partial derivatives. Let $\underaccent{\tilde}{g}=\left(g(\bm{\theta}_1),g(\bm{\theta}_2),\ldots,g(\bm{\theta}_K)\right)$ and define\\
$
\Omega=\left\{\underaccent{\tilde}{g}=\left(g(\bm{\theta}_1),g(\bm{\theta}_2),\ldots,g(\bm{\theta}_K)\right): g(\bm{\theta}_i)\in\mathbb{R},\ i=1,2,\ldots,K \right\}
$
as the parameter space corresponding to $\underaccent{\tilde}{g}$. We have to identify the extreme population among the given populations using a generalized measure represented by $g\left(\bm{\theta}_i\right)$. Now, we define which population is considered as extreme population. An extreme population may correspond to either the best or the worst population. If the best population is defined as the one having the minimum value of the measure, then the worst population corresponds to the one having the maximum value. By considering the negative of the measure, the roles of the best and worst populations are reversed. A similar argument applies when the worst population is defined as the one having the minimum value of the measure. Therefore, without loss of generality, we define the extreme population as the population having the minimum value of $g\left(\bm{\theta}_i\right)$. Further, it may not be practically meaningful to distinguish between two populations whose values of $g\left(\bm{\theta}_i\right)$ are arbitrarily close. Thus, given $\delta^*>0$, our objective is to identify the extreme population if the extreme population is separated from the second-best (or second-worst) population by at least $\delta^*$. Hence,
\begin{eqnarray}\label{3.best}
j^{th}\text{ population is extreme if }g\left(\bm{\theta}_j\right)=\min\limits_{i=1,2,\ldots,K}g\left(\bm{\theta}_i\right),
\end{eqnarray}
whenever the true parameter configuration belongs to the preference zone $\Omega(\delta^*)$, where
\begin{eqnarray}\label{3.preferencezone}
\Omega(\delta^*)=\left\{(g(\bm\theta_1),g(\bm\theta_2),\ldots,g(\bm\theta_K))\; ; \;g\left(\bm\theta_{[2]}\right)-g\left(\bm\theta_{[1]}\right)\geq \delta^*\right\}.
\end{eqnarray}
The quantity $g\left(\bm\theta_{[1]}\right)$ represents the value of measure corresponding to extreme population, while $g\left(\bm\theta_{[2]}\right)$ represents the value of measure corresponding to second-best (or second-worst) population. In general, the quantity $g\left(\bm\theta_{[i]}\right)$ represents the value of measure corresponding to $i^{th}$ best (or $i^{th}$ worst) population.
The complement of $\Omega(\delta^*)$, denoted by $\Omega^C(\delta^*) \;(=\Omega\setminus \Omega(\delta^*))$, is referred to as the indifference zone. For more details related to preference and indifference zone, see \cite{bechhofer1954single}.
\subsection{Multistage sampling design}\label{3.S2}
Consider the $i^{th}$ population, $i=1,2,\dots,K$, consisting of $S_i$ strata indexed by $s_i=1,2,\dots,S_i$. The $s_i^{th}$ stratum contains $H_{is_i}$ primary sampling units (PSUs), denoted by $c_{is_i}=1,2,\dots,H_{is_i}$. Thus the total number of PSUs in population $i$ is $H_i=\sum_{s_i=1}^{S_i}H_{is_i}$. Further, each PSU is divided into $L$ sub-strata indexed by $b_{ic_{is_i}}=1,2,\ldots,L$. The $b_{ic_{is_i}}^{th}$ sub-stratum contains $M_{s_ic_{is_i}b_{ic_{is_i}}}$ secondary sampling units (SSUs), which are the ultimate sampling units. Hence, the total number of SSUs in the $c_{is_i}^{th}$ PSU is $M_{s_ic_{is_i}}=\sum_{b_{ic_{is_i}}=1}^{L}M_{s_ic_{is_i}b_{ic_{is_i}}}$.


For the purpose of estimation, $n_{is_i}$ PSUs are drawn from the $s_i^{\text{th}}$ stratum according to a probability proportional to size sampling scheme with replacement, taking the number of SSUs as the size measure. Thereafter, $k$ SSUs are independently selected from each sub-stratum of every sampled PSU using simple random sampling. The total number of sampled PSUs is
\[n_i = \sum_{s_i=1}^{S_i} n_{is_i}, \quad \text{with} \quad \frac{n_{is_i}}{n_i} = \frac{H_{is_i}}{H_i} = a_{is_i}.
\]

Let $x_{s_ic_{is_i}b_{ic_{is_i}}h}$ denote the observed value of the study variable for the $h^{\text{th}}$ SSU belonging to the $b_{ic_{is_i}}^{\text{th}}$ sub-stratum of the $c_{is_i}^{\text{th}}$ PSU in stratum $s_i$. In multistage sampling designs, sampling weights are assigned to ensure that the collected sample appropriately represents the target population by taking unequal selection probabilities into account. Accordingly, each sampled SSU is assigned a weight $W_{s_ic_{is_i}b_{ic_{is_i}}h}$, defined as the reciprocal of its inclusion probability. The inclusion probability of a sampled SSU is obtained as the product of the probability of selecting PSU $c_{is_i}$ from stratum $s_i$, denoted by $p_{c_{is_i}}$ and the probability of selecting the $h^{\text{th}}$ SSU from sub-stratum $b_{ic_{is_i}}$ within the selected PSU, denoted by $p_{b_{ic_{is_i}}h}$. Accordingly, under probability proportional to size sampling with replacement at the PSU stage and simple random sampling within sub-strata at the SSU stage, the normalized sampling weight is defined as:
\[w_{s_ic_{is_i}b_{ic_{is_i}}h}=\frac{W_{s_ic_{is_i}b_{ic_{is_i}}h}}{\sum_{s_i=1}^{S_i}\sum_{c_{is_i}=1}^{n_{is_i}}\sum_{b_{ic_{is_i}}=1}^{L}\sum_{h=1}^{k}W_{s_ic_{is_i}b_{ic_{is_i}}h}},\text{ and }
W_{s_ic_{is_i}b_{ic_{is_i}}h}
=\frac{\left(\sum_{c_{is_i}=1}^{H_{is_i}} M_{s_ic_{is_i}}\right)
M_{s_ic_{is_i}b_{ic_{is_i}}}}
{n_{is_i}M_{s_ic_{is_i}}k}.
\]
\subsection{Generalized Method of Moments-Based Estimation of the Measure of Interest and Its Asymptotic Properties}\label{3.S2.2}
For $i=1,2,\dots,K$, let $X_i$ denote the characteristic of an individual of $i^{th}$ population. The probability density function of characteristic $X_i$ in the $s_i^{th}$ stratum is denoted by $F(x_i|s_i)$. We consider the problem of estimating an $l$-dimensional parameter vector $\bm{\theta}_i=(\theta_{i1},\theta_{i2},\dots,\theta_{il})$, which satisfies the following population moment condition:\[\sum_{s_i=1}^{S_i}\int H_{is_i}\,\textbf{m}_i(x_i,\bm{\theta}_i)\,dF(x_i|s_i)=\bm{0},\]
where $ \textbf{m}_i(x_i,\bm{\theta}_i)=\left(m_i^1(x_i,\bm{\theta}_i),m_i^2(x_i,\bm{\theta}_i),\ldots,m_i^l(x_i,\bm{\theta}_i)\right)$
is an $l$-dimensional vector of moment functions. The corresponding sample version of the above population moment condition is given by
\[\sum_{s_i=1}^{S_i}\sum_{c_{is_i}=1}^{n_{is_i}}\sum_{b_{ic_{is_i}}=1}^{L}\sum_{h=1}^{k}w_{s_ic_{is_i}b_{ic_{is_i}}h}\textbf{m}_i(x_{s_ic_{is_i}b_{ic_{is_i}}h},\bm{\theta})=\bm{0}.\]
For convenience, the above sample moment condition can be rewritten in terms of PSU-level contributions as
$\frac{1}{n_i}\sum_{j=1}^{n_i}\widetilde{\textbf{m}}_{ij}(\bm{\theta})=\bm{0},$
where
\[\widetilde{\textbf{m}}_{ij}(\bm{\theta})=\sum_{s_i=1}^{S_i}\frac{1}{a_{is_i}}\textbf{1}(s_r=s_i)\frac{\left(\sum_{c_{js_r}=1}^{H_{js_r}}M_{s_r{c_{js_r}}}\right)}{M_{s_rj}k}\sum_{b_{jc_{js_r}}=1}^{L}M_{s_r{j}b_{jc_{js_r}}}\sum_{h=1}^{k}\textbf{m}_i(x_{s_rc_{js_r}b_{jc_{js_r}}h},\bm{\theta}_i).\]
Consider a function $f_{in_i}(\bm{\theta})=\frac{1}{n_i}\sum_{j=1}^{n_i}\widetilde{\textbf{m}}_{ij}(\bm{\theta}).$
The generalized method of moments (GMM) estimator $\widehat{\bm{\theta}}_{in_i}$ for the parameter $\bm{\theta}_i$ is defined as the value of $\bm{\theta}$ that minimizes the objective function $f_{in_i}(\bm{\theta})'A_{n_i}f_{in_i}(\bm{\theta}),$
that is,
$\widehat{\bm{\theta}}_{in_i}=\arg\min\left\{f_{in_i}(\bm{\theta})'A_{n_i}f_{in_i}(\bm{\theta})\right\},$
where $A_{in_i}$ is positive definite matrix with probability $1$ and $A_{in_i}\xrightarrow{a.s.}A_{i0}$ for each $n_i$.

Now, we establish asymptotic properties of generalized method of moment estimator of the parameter $\bm{\theta}_i$. For $i=1,2,\ldots,K$, under the assumptions \ref{3.AN1}-\ref{3.AN16} (in Appendix), using \cite{BD2005} and assuming that populations are independent to each other and $n_{is_i}\rightarrow\infty$ for each $s_i=1,2,..,S_i$ at the same rate, we obtain the following results:
$\widehat{\bm\theta}_{in_i}\xrightarrow{P}\bm{\theta}_i,$
and 
\begin{eqnarray}\label{3.2}
\sqrt{n_i}\left(\widehat{\bm\theta}_{in_i}-\bm{\theta}_i\right)\xrightarrow{d}N\left(0,\Sigma_i\right),
\end{eqnarray}
where $\Sigma_i$ denotes the asymptotic covariance matrix of $\sqrt{n_i}\left(\widehat{\bm{\theta}}_{in_i}\right)$.

Now, we proceed to find the estimator of function $g(\bm\theta_i)$ introduced in Subsection \ref{3.S2.1}.
Since $g(\cdot)$ is continuous and $\bm{\widehat{\theta}}_{in_i}$ is a consistent GMM estimator of $\bm{\theta}_i$, it follows from the continuous mapping theorem that
\begin{eqnarray}\label{3.3}
g\left(\bm{\widehat{\theta}}_{in_i}\right)\xrightarrow{P}g\left(\bm{\theta}_i\right).
\end{eqnarray}
Hence, $g\left(\bm{\widehat{\theta}}_{in_i}\right)$ is a consistent estimator of $g(\bm{\theta}_i)$. For asymptotic distribution, in Lemma \ref{3.L1}, we establish that
\begin{eqnarray}\label{3.4}
\sqrt{n_i}\left(g\left(\bm{\widehat{\theta}}_{in_i}\right)-g\left(\bm{\theta}_i\right)\right)\xrightarrow{d}N\left(0,\xi_i^2\right),
\end{eqnarray}
where $\xi_i^2$ denotes the asymptotic variance of $\sqrt{n_i}\left(g\left(\bm{\widehat{\theta}}_{in_i}\right)\right)$. These asymptotic results are useful for constructing confidence intervals and performing hypothesis testing for the true quantity $g\left(\bm{\theta}_i\right)$. In particular, when the asymptotic variance $\xi_i^2$ is known, an asymptotic $100(1-\alpha)\%$ confidence interval for $g\left(\bm{\theta}_i\right)$ is given by
\[\left(g\left(\bm{\widehat{\theta}}_{in_i}\right)-z_{\alpha/2}\frac{\xi_i}{\sqrt{n_i}},\;g\left(\bm{\widehat{\theta}}_{in_i}\right)+z_{\alpha/2}\frac{\xi_i}{\sqrt{n_i}}\right).\]


\noindent From the asymptotic distribution, it follows that for sufficiently large $n_i$, the variance of $g\left(\bm{\widehat{\theta}}_{in_i}\right)$ is $\xi_i^2/n_i$. As a result, estimators corresponding to different populations have unequal variances. Direct comparison based on these estimators may produce misleading conclusions. For instance, consider two populations satisfying $g\left(\bm{\theta}_1\right)<g\left(\bm{\theta}_2\right),$ with the difference $g\left(\bm{\theta}_2\right)-g\left(\bm{\theta}_1\right)$ being small. Let the population corresponding to $g\left(\bm{\theta}_1\right)$ as extreme population. Since the true values are unknown, the selection of the extreme population is carried out using the estimators $g\left(\bm{\widehat{\theta}}_{1n_1}\right)$ and $g\left(\bm{\widehat{\theta}}_{2n_2}\right)$. Now, if the variance of estimator corresponding to population $1$ is larger as compared to variance of estimator corresponding to population $2$, there is high probability that $g\left(\bm{\widehat{\theta}}_{1n_1}\right)$ will exceed $g\left(\bm{\widehat{\theta}}_{2n_2}\right)$. As a result, an incorrect population may be selected as the extreme population. Therefore, the extreme population should be selected on the basis of estimators having equal variances. To maintain a common variance across populations, the sample PSU sizes $n_1,n_2,\ldots,n_K$ are chosen such that
$\xi_1^2/n_1=\xi_2^2/n_2=\cdots=\xi_K^2/n_K=t,$
where $t$ is a positive constant. Consequently, for $i=1,2,\ldots,K$, we get
\begin{eqnarray}\label{3.samplesize}
n_i=\frac{\xi_i^2}{t}.
\end{eqnarray}
If the asymptotic variances $\xi_1^2,\xi_2^2,\ldots,\xi_K^2$ are known, then the estimators $g\left(\bm{\widehat{\theta}}_{in_i}\right)$ may be directly used for selecting the extreme population, where $n_i$ satisfy the equation \eqref{3.samplesize}. However, in practice, the asymptotic variances $\xi_1^2,\xi_2^2,\ldots,\xi_K^2$ are generally unknown and therefore consistent estimators of these quantities are required. Suppose $V_{in_i}^2$ is a consistent estimator of $\xi_i^2$. Using this consistent estimator, the required sample PSU size is determined on the basis of equation \eqref{3.samplesize} as follows: choose an integer $n_i$ satisfying
\begin{eqnarray}\label{3.t}
n_i \geq \frac{V_{in_i}^2}{t}.
\end{eqnarray}
For such large $n_i$, the estimator of variance of the estimator $g\left(\bm{\widehat{\theta}}_{in_i}\right)$ still differ across populations. Hence, these estimators are not suitable for extreme population selection. Motivated by this limitation, we construct an alternative estimator of $g\left(\bm{\theta}_i\right)$  whose estimated variance remains the same across populations. Suppose $n_0$ is the pilot sample PSU size and $n_i$ is an integer such that $n_i>n_0$ and satisfy equation \eqref{3.t}. We define another estimator of $g\left(\bm{\theta}_i\right)$ as follows: 
\begin{eqnarray}\label{3.gestimator}
\widetilde{g}\left(\widehat{\bm{\theta}}_{in_0},\widehat{\bm{\theta}}_{i(n_{i}-n_0)}\right)=c_{i} g\left(\widehat{\bm{\theta}}_{in_0}\right)+(1-c_{i})g\left(\widehat{\bm{\theta}}_{i\left(n_{i}-n_0\right)}\right),
\end{eqnarray}
where
$c_{i}=\left(n_0/n_{i}\right)\left(1+\sqrt{1-\left(n_{i}/n_0\right)\left(1-\left(\left(n_{i}-n_0\right)t/V_{in_{i}}^2\right)\right)}\right).$
This value of $c_i$ ensures that the estimator of variance of $\widetilde{g}\left(\widehat{\bm{\theta}}_{in_0},\widehat{\bm{\theta}}_{i(n_i-n_0)}\right)$ is constant for all $i$. The justification is provided below.

From Lemma \ref{3.L3}, for sufficiently large $n_0$ and $n_i$, the variance of the estimator $\widetilde{g}\left(\widehat{\bm{\theta}}_{in_0},\widehat{\bm{\theta}}_{i(n_i-n_0)}\right)$ is given by
$\left(c_i^2\xi_i^2/n_0\right)+\left((1-c_i)^2\xi_i^2/(n_i-n_0)\right).$
Since $V_{in_i}^2$ is a consistent estimator of $\xi_i^2$, an estimator of variance is
$\left(c_i^2V_{in_i}^2/n_0\right)+\left((1-c_i)^2V_{in_i}^2/\left(n_i-n_0\right)\right).$
Using the value of $c_i$, we get
$\left(c_i^2V_{in_i}^2/n_0\right)+\left((1-c_i)^2V_{in_i}^2/\left(n_i-n_0\right)\right)=t,$
where $t$ is a positive constant. As a result, estimator of variance of $\widetilde{g}\left(\widehat{\bm{\theta}}_{in_0},\widehat{\bm{\theta}}_{i(n_i-n_0)}\right)$ is constant for all $i$. 

The sample PSU size which satisfy equation \eqref{3.t} guarantee that the estimator of variance of $\widetilde{g}\left(\widehat{\bm{\theta}}_{in_0},\widehat{\bm{\theta}}_{i(n_i-n_0)}\right)$ is constant for all $i$. However, it does not guarantee that direct comparison on the basis of this estimator correctly selects the true extreme population.  Motivated by this, we develop two procedures for extreme population selection in Section \ref{3.S3} based on equation \eqref{3.t} and the estimator $\widetilde{g}\left(\widehat{\bm{\theta}}_{in_0},\widehat{\bm{\theta}}_{i(n_i-n_0)}\right)$.

\subsection{Illustration}\label{3.S2.3}
In this subsection, we show that several commonly used measures arise as special cases of the proposed generalized measure. In particular, we demonstrate this by showing that the Gini index and a TMB score based measure can be obtained as special cases of the proposed framework.
\subsubsection{Illustration using Gini index}\label{3.S2.3.2}
Economic inequality within a country or state mainly arises because income is not distributed equally among individuals. Understanding the extent of such inequality is important for policymakers when designing economic and social welfare policies. A variety of measures have been proposed in literature to quantify inequality, among which the Gini index has emerged as one of the most commonly used measures for evaluating income inequality within a population. Identifying the population with the minimum Gini index is important because it corresponds to the population with the most equitable distribution of income or wealth. Such information can assist policymakers and researchers in evaluating economic equality, assessing the effectiveness of social policies and identifying practices that promote a more balanced distribution of resources. Motivated by this, our objective is to identify the population having minimum value of Gini indices among $K$ independent populations.

For $i=1,2,\ldots,K$, the Gini index can be expressed as
$
G_{X_i} = 1 - 2\int_{0}^{1} \phi_i(p)\, dp,
$
where $\phi_i(\cdot)$ denotes the Lorenz curve defined by
$
\phi_i(p) = \alpha_i(p)/\mu_i,\; 
\alpha_i(p) = \int_{0}^{p} z_i(u)\, du,
$
and
$
z_i(u) = \inf\{x : u \le F_i(x)\}, \; u \in [0,1],
$
represents the $u$-th population quantile, with $F_i(\cdot)$ being the cumulative distribution function of the random variable $X_i$.
We can write $G_{X_i}$ as 
$G_{X_i}=g(\bm{\theta}_i),$ where $\bm{\theta}_i=\left(\theta_{i1},\theta_{i2}\right)$, $\theta_{i1}=\int_{0}^{1}\alpha_i(p)dp,\;\theta_{i2}=\mu_i,$ and $g(\bm{\theta}):\mathbb{R}^2\rightarrow\mathbb{R}$ is continuous first-order partial derivative function defined as follows:
$g\left(\theta_{i1},\theta_{i2}\right)=1-\left(2\theta_{i1}/\theta_{i2}\right),\; \theta_{i2}\neq0.$
This is the same as the function discussed in Subsection \ref{3.S2.1}.

Since the true value of the Gini index is unknown, it must be estimated from sample data. Accurate estimation requires an appropriate sample that adequately represents the characteristics of the population. To obtain a representative sample, a multistage sampling framework is employed for the $i^{th}$ population. At first, the $i^{th}$ population is divided into strata. Strata are formed based on geographical or residential categories, such as rural and urban areas. Within each stratum, clusters are selected, for example, villages in rural regions and blocks in urban regions. Further, substrata may be created on the basis of socioeconomic characteristics, such as affluent and non-affluent households. Each substratum then consists of the sampled household units. This multistage design is the same as discussed in Section \ref{3.S2}. 

Now, we proceed to find the estimator of Gini indices. For $i=1,2,...,K$, the estimator of Gini index $G_{X_i}$ is given by
\begin{eqnarray}
\widehat{G}_{in_i} = g\left(\widehat{\theta}_{i1n_i},\widehat{\theta}_{i2n_i}\right)=1-\frac{2\widehat{\theta}_{i1n_i}}{\widehat{\theta}_{i2n_i}},
\end{eqnarray}
where
$\widehat{\theta}_{i1n_i}=\int_{0}^{1}\widehat{\alpha}_{in_i}(p)dp,\;\; \widehat{\alpha}_{in_i}(p)=\int_{0}^{p}\widehat{z}_{in_i}(u)du,$
and
$\widehat{\theta}_{i2n_i}=\sum_{s_i=1}^{S_i}\sum_{c_{is_i}=1}^{n_{is_i}}\sum_{b_{ic_{is_i}}=1}^{2}\sum_{h=1}^{k}w_{s_ic_{is_i}b_{ic_{is_i}}h}x_{s_ic_{is_i}b_{ic_{is_i}}h}.$
Here, for $p\in[0,1]$, 
$\widehat{z}_{in_i}(p)=\inf\limits_{s_i,c_{is_i},b_{ic_{is_i}},h}\{x_{s_ic_{is_i}b_{ic_{is_i}}h}; \widehat{F}_n(x_{s_ic_{is_i}b_{ic_{is_i}}h})\geq p\},$
\[\widehat{F}_{in_i}(x)=\sum_{a_i=1}^{S_i}\sum_{d_i=1}^{n_{ia_i}}\sum_{r_{id_i}=1}^{2}\sum_{c=1}^{k}w_{a_id_ir_{id_i}c}\textbf{1}(x_{a_id_ir_{id_i}c}\leq x),\]
$w_{s_ic_{is_i}b_{ic_{is_i}}h}$ are the same as defined in Section \ref{3.S2} and $x_{s_ic_{is_i}b_{ic_{is_i}}h}$ denote the observed value of the study variable for the $h^{\text{th}}$ household belonging to sub-stratum $b_{ic_{is_i}}$ within cluster $c_{is_i}$ of stratum $s_i$ (e.g., household income or expenditure).
For $i=1,2,...,K$, estimator of Gini index $G_{X_i}$ can be expressed as follows:
\begin{eqnarray}\label{3.Gini}
\widehat{G}_{in_i} = 1 - \frac{2}{\widehat{\mu}_{in_i}} \sum_{s_i=1}^{S_i} \sum_{c_{is_i}=1}^{n_{is_i}}\sum_{b_{ic_{is_i}}=1}^{2} \sum_{h=1}^{k}w_{s_ic_{is_i}b_{ic_{is_i}}h} x_{s_ic_{is_i}b_{ic_{is_i}}h} \left(1 - \widehat{F}_{in_i}(x_{s_ic_{is_i}b_{ic_{is_i}}h})\right).
\end{eqnarray}
Under the assumptions \ref{3.AN1}-\ref{3.AN16} (in Appendix), following \cite{BD2007}, as $n_{is_i} \to \infty$ at the same rate for all strata, we have
$\sqrt{n_i}\left(\widehat{G}_{in_i}-G_{X_i}\right)\xrightarrow{d}N(0,\xi_i^2),$
where $\xi_i^2$ denotes the asymptotic variance of $\sqrt{n_i}\widehat{G}_{in_i}$. 

Since the income distribution of the population is unknown, the value of $\xi_i^2$ is also unknown. To this end, we require a consistent estimator of $\xi_i^2$, $i=1,2,...,K$. Following the method of variance estimation proposed by \cite{Binder}, we employ an approach that is consistent, theoretically rigorous and computationally efficient; specifically, for $i=1,2,...,K$, a consistent estimator of $\xi_i^2$ is given by
\begin{eqnarray}\label{3.VARIANCEGINI}
V_{in_i}^2=\sum_{s_i=1}^{S_i}\frac{n_i*n_{is_i}}{n_{is_i}-1}\sum_{c_{is_i}=1}^{n_{is_i}}\left(u_{is_ic_{is_i}}-\bar{u}_{is_i}\right)^2, 
\end{eqnarray}
where
$\bar{u}_{is_i}=\frac{1}{n_{is_i}}\sum_{c_{is_i}=1}^{n_{is_i}}u_{is_ic_{is_i}},$
and
\begin{eqnarray*}
u_{is_ic_{is_i}}&=&\frac{2}{\widehat{\mu}_{in_i}}\sum_{b_{ic_{is_i}}=1}^{2}\sum_{h=1}^{k}w_{s_ic_{is_i}b_{c_{is_i}}h}\bigg[\left(\widehat{F}(x_{s_ic_{is_i}b_{c_{is_i}}h})-\frac{\widehat{G}_{in_i}+1}{2}\right)x_{s_ic_{is_i}b_{c_{is_i}}h}
\\&&+\sum_{a_i=1}^{S_i}\sum_{d_i=1}^{n_{ia_i}}\sum_{r_{id_i}=1}^{2}\sum_{c=1}^{k}w_{a_id_ir_{id_i}c}x_{a_id_ir_{id_i}c}\textbf{1}(x_{a_id_ir_{id_i}c}\geq x_{s_ic_{is_i}b_{c_{is_i}}h})-\frac{\widehat{\mu}_{in_i}}{2}(\widehat{G}_{in_i}+1)\bigg].
\end{eqnarray*}
From the above discussion, it follows that both the estimator of the Gini index and its corresponding variance estimator satisfy the general framework established in Subsection \ref{3.S2.2}. The another estimator, analogous to the estimator defined in equation \eqref{3.gestimator}, is given by
\begin{eqnarray}\label{3.gGINI}
\widetilde{G}_{in_i}=c_i\widehat{G}_{in_0}+(1-c_i)\widehat{G}_{i(n_i-n_0)},
\end{eqnarray}
where
$c_{i}=\left(n_0/n_{i}\right)\left(1+\sqrt{1-\left(n_{i}/n_0\right)\left(1-\left(\left(n_{i}-n_0\right)t/V_{in_{i}}^2\right)\right)}\right).$
\subsubsection{Illustration using TMB score in genetics}\label{3.S2.3.1}
As discussed in Section \ref{3.S1}, tumor mutation burden (TMB) is a widely used genomic biomarker for predicting response to immune checkpoint inhibitor (ICI) therapy, with higher TMB generally associated with improved survival outcomes \citep{gandara2025tumor, sorich2025tumour}.
TMB score quantifies the number of somatic mutations per megabase in the coding regions of a tumour genome \citep{chalmers2017analysis}. A high TMB is thought to increase the generation of neoantigens, which can enhance immune recognition of tumour cells and thereby improve the likelihood of a favorable response to immunotherapy. Identifying populations with a higher concentration of low-TMB patients is therefore useful for guiding treatment choice, and requires an appropriate summary measure for comparison across populations.

Existing studies have compared TMB across populations using different summary measures. For example, \cite{luo2025ancestral} compared TMB after adjusting for covariates such as age, gender, cancer type, and tumor stage. \cite{chalmers2017analysis} used the logarithm of TMB scores to study the relationship between gene alterations and tumor mutation burden. For the illustration of our procedures in Application section, we use the mean of the log-transformed TMB score as our measure of interest. The logarithmic transformation reduces the right skewness typical of TMB distributions, and a smaller mean of log (TMB) indicates a population with a greater concentration of low-TMB individuals. This makes it an interpretable measure for identifying the population of interest. Since TMB is also known to vary by age, gender, and cancer type, each population is stratified accordingly.

Suppose there are $K$ independent populations, indexed by $i\in\{1,2,\ldots,K\}$. Since TMB values can differ across age groups, genders and cancer types, each population is stratified based on one of these characteristics. Suppose $i^{th}$ population is divided in $S_i$ strata, indexed by $s_i=1,2,\ldots,S_i$.  Here, we do not further subdivide the population into additional stages; rather, individuals are sampled directly from each stratum. This multistage design is a special case of design discussed in Subsection \ref{3.S2} by setting $L=1$ and $M_{s_ic_{is_i}}=1$. Total number PSUs becomes the total number of individuals in the dataset and PSUs becomes our ultimate sampling units. In each stratum, we have $H_{is_i}$ individuals, indexed by $c_{is_i}=1,2,\ldots,H_{is_i}$. In practice, the total number of individuals $H_{is_i}$ can be obtained from national and international cancer surveillance databases, including the International Agency for Research on Cancer (IARC)\footnote{\url{https://gco.iarc.who.int/en}}, the National Cancer Registry Programme (NCRP)\footnote{\url{https://ncdirindia.org/All_Reports/Report_2020/default.aspx}} and the SEER Program\footnote{\url{https://seer.cancer.gov/}}, among others. Let $x_{s_ic_{is_i}}$ denote the TMB value of $c_{is_i}^{th}$ individual from $s_i^{th}$ stratum. For $i=1,2,\ldots,K$, the average of log-transformed TMB scores is given by
\begin{eqnarray}\label{3.TMB}
\mu_{TMB_i}=\frac{1}{H_i}\sum_{s_i=1}^{S_i}\sum_{c_{is_i=1}}^{H_{is_i}}log\left(x_{s_ic_{is_i}}+1\right),
\end{eqnarray}
where $H_i=\sum_{s_i=1}^{S_i}H_{is_i}$. We can write $\mu_{TMB_i}$ as 
$\mu_{TMB_i}=g(\theta_i),$
where $\theta_i=\frac{1}{H_i}\sum_{s_i=1}^{S_i}\sum_{c_{is_i=1}}^{H_{is_i}}log\left(x_{s_ic_{is_i}}+1\right)$ and $g(\bm{\theta}_i):\mathbb{R}\rightarrow\mathbb{R}$ is continuous first-order partial derivative function defined as follows: $g(\theta_i)=\theta_i.$
This matches with the function discussed in Subsection \ref{3.S2.1}.

Now, we proceed to find the estimator of mean of log-transformed TMB scores. For estimation of $\mu_{TMB_i}$, we sample $n_{is_i}$ individuals from $s_i^{th}$ stratum. The normalized sampling weight is given by 
$w_{s_ic_{is_i}}=W_{s_ic_{is_i}}/\sum_{s_i=1}^{S_i}\sum_{c_{is_i}=1}^{n_{is_i}}W_{s_ic_{is_i}},$
$W_{s_ic_{is_i}}=H_{is_i}/n_{is_i}.$
Using the normalized weight, the GMM estimator of $\mu_{TMB_i}$ is given by
\[\widehat{\mu}_{TMB_{in_i}}=\sum_{s_i=1}^{S_i}\sum_{c_{is_i=1}}^{n_{is_i}}w_{s_ic_{is_i}}log\left(x_{s_ic_{is_i}}+1\right).\]
Consequently, from equation \eqref{3.2}, for $i=1,2,..,K$, we have
$\sqrt{n_i}\left(\widehat{\mu}_{TMB_{in_i}}-\mu_{TMB_i}\right)\xrightarrow{d}N\left(0,\xi_i^2\right),$
where $\xi_i^2$ denotes the asymptotic variance of $\sqrt{n_i}\left(\widehat{\mu}_{TMB_{in_i}}\right)$. The value of $\xi_i^2$ is usually unknown. To this end, we require a consistent estimator of $\xi_i^2$, $i=1,2,...,K$. Along the lines of \cite{Binder}, we employ an approach that is consistent, theoretically rigorous and computationally efficient; specifically, for $i=1,2,...,K$, a consistent estimator of $\xi_i^2$ is given by
\begin{eqnarray}\label{3.VARIANCETMB}
V_{in_i}^2=\sum_{s_i=1}^{S_i}\frac{n_i*n_{is_i}}{n_{is_i}-1}\sum_{c_{is_i}=1}^{n_{is_i}}\left(u_{is_ic_{is_i}}-\bar{u}_{is_i}\right)^2, 
\end{eqnarray}
where $u_{is_ic_{is_i}}=w_{s_ic_{is_i}}log\left(x_{s_ic_{is_i}}+1\right),$ and $\bar{u}_{is_i}=\frac{1}{n_{is_i}}\sum_{c_{is_i}=1}^{n_{is_i}}u_{is_ic_{is_i}}.$

From the above discussion, it follows that both the estimator of $\mu_{TMB_i}$ and its corresponding variance estimator satisfy the general framework established in Subsection \ref{3.S2.2}. The another estimator, analogous to the estimator defined in equation \eqref{3.gestimator}, is given by
\begin{eqnarray}\label{3.gTMB}
\widetilde{\mu}_{TMB_{in_i}}=c_i\widehat{\mu}_{TMB_{in_0}}+(1-c_i)\widehat{\mu}_{TMB_{i(n_i-n_0)}},
\end{eqnarray}
where
$c_{i}=\left(n_0/n_{i}\right)\left(1+\sqrt{1-\left(n_{i}/n_0\right)\left(1-\left(\left(n_{i}-n_0\right)t/V_{in_{i}}^2\right)\right)}\right).$
Since both the estimator $\widehat{\mu}_{TMB_{in_i}}$ and its variance estimator satisfy the general framework of Subsection~\ref{3.S2.1}, the proposed online and MAB-based algorithms of Section \ref{3.S3} can be applied directly to identify the population with the minimum mean $\log(TMB)$ score. We illustrate this application using a real clinical sequencing cohort in Section~\ref{3.S5}.

The generalized measure introduced in this section provides a unified framework for characterizing a wide class of population measures. The examples discussed above illustrate its ability to encompass measures from different fields. This generality enables the development of selection procedures that are applicable beyond any specific measure. Therefore, in the next section, we develop two algorithms for selecting the extreme population based on the generalized measure.
\section{Algorithms for selecting the extreme population}\label{3.S3}
Since the true values of $g\left(\bm{\theta}_i\right)$ are unknown, we cannot identify the true extreme population. Therefore, the extreme population must be identified using estimated values based on sampled data. In this section, we develop two sequential algorithms using the estimator defined in \eqref{3.gestimator} for identifying the extreme population. Both algorithms are constructed on the basis of positive constant $t$. The first is an online algorithm in which the value of $t$ is fixed in advance and PSUs are collected sequentially according to a prescribed criterion. The second algorithm is formulated within the Multi-Armed Bandit (MAB) framework, where the value of $t$ is updated at each stages and based on the updated value, PSUs are sequentially sampled from the populations during each stages. Both procedures aim to identify the extreme population among the given populations.
\subsection{Online algorithm}\label{3.S3.1}
For each $i=1,2,\ldots,K$, the online algorithm starts by selecting $n_{0s_i}$ PSUs from every stratum $s_i$ of the $i^{th}$ population. Thus, the total number of PSUs selected from the $i^{th}$ population during the pilot stage is $n_{0}=\sum_{s_i=1}^{S_i}n_{0s_i}$. From each sub-stratum of the selected PSUs, $k$ SSUs are randomly sampled. Using the observations collected during the pilot stage, an estimator $V_{in_{0}}^2$ of the asymptotic variance $\xi_i^2$ is computed for the $i^{th}$ population. Sampling then continues sequentially according to a stopping rule. Specifically, for a given $l>0$, let $N_i(\leq H_i)$ denote the smallest integer $n_i(\geq n_{0})$, for which
\begin{eqnarray}\label{3.20}
n_{i}\geq\frac{1}{t}\left(V_{in_{i}}^2+\frac{1}{n_{i}^l}\right)=\widehat{C}_{i} \text{ and } n_{is_i}\geq\widehat{C}_{is_i}=\widehat{C}_{i}a_{is_i}, \text{ for all } s_i=1,2,..,S_i,
\end{eqnarray}
where $t=\left(\frac{\delta^*}{h}\right)^2$ and $h$ satisfies the following equation:
$\int_{-\infty}^{\infty}\left(1-\Phi\left(y-h\right)\right)^{K-1}\phi(y)dy=1-\delta$
and  $0<\delta<1$. The term $1/n_i^l$ is added to prevent the procedure from stopping too early. In the absence of this term, the stopping criterion in \eqref{3.20} may be satisfied prematurely when only a few observations have been collected because $V_{in_i}^2$ can take small values during the early stages of sampling. In addition, $V_{in_i}^2 + 1/n_i^l$ remains a consistent estimator of $\xi_i^2$.

If the stopping criterion is satisfied, the procedure terminates and the current pilot sample size is taken as the final sample size. Otherwise, an additional $m'(\geq1)$ PSUs are selected from each stratum satisfying $n_{is_i}<\widehat{C}_{is_i}$ and from each newly selected PSU, $k$ SSUs are randomly sampled. The asymptotic variance $\xi_i^2$ is then re-estimated using the updated data and the stopping rule is checked again. This process is repeated until the stopping rule is satisfied. The resulting final PSU size from stratum $s_i$ is given by $N_{is_i}=N_ia_{is_i}$. Using the final PSU size $N_i$, the estimator of $g\left(\bm{\theta}_i\right)$ is taken as $\widetilde{g}\left(\widehat{\bm{\theta}}_{in_0},\widehat{\bm{\theta}}_{i(N_i-n_0)}\right)$. Based on these estimators, the extreme population selection rule is defined as follows: $j^{th}$ population is declared as the extreme population if 
$
\widetilde{g}\left(\widehat{\bm{\theta}}_{jn_0},\widehat{\bm{\theta}}_{j(N_j-n_0)}\right)=\min\limits_{i=1,2,\ldots,K}\widetilde{g}\left(\widehat{\bm{\theta}}_{in_0},\widehat{\bm{\theta}}_{i(N_i-n_0)}\right).
$
\subsection{MAB algorithm}\label{3.S3.2}
A Multi-Armed Bandit (MAB) framework is a sequential decision-making approach in which several competing populations, called arms, are sampled adaptively over different stages. At each stage, the information collected so far is used to decide which populations should receive additional samples and which populations can be eliminated from further consideration. Thus, the MAB framework provides an efficient way to allocate samples sequentially and reduce unnecessary sampling from inferior populations. On the basis this framework, we develop the following algorithm.

For the populations considered in this algorithm, no two true values of $g(\bm{\theta}_i)$ are equal. Thus, for $i\neq j$, $g(\bm{\theta}_i)\neq g(\bm{\theta}_j)$. For each $i=1,2,\ldots,K$, the procedure begins with the selection of $n_{0s_i}$ PSUs from every stratum $s_i$ of the $i^{th}$ population. Thus, the total number of PSUs sampled from $i^{th}$ population at pilot stage is $n_{0}=\sum_{s_i=1}^{S_i}n_{0s_i}$. From each sub-stratum within the selected PSU, we randomly select $k$ SSUs. Using the pilot sample data, we compute $V_{in_0}^2$.
Sampling then proceeds in following stages, indexed by $r$. At each stage, the sampling procedure is governed by a stopping rule, which vary from one stage to another.\\ 
\textit{Stage $r=1:$ } In this stage, the sampling continues according to a stopping rule. Specifically, for a given $l>0$ and $\alpha>0$, let $N_{i1}(\leq H_i)$ is the smallest integer $n_{i1}(\geq n_{0})$, for which
\begin{eqnarray}\label{3.21}
n_{i1}\geq\frac{1}{t_1}\left(V_{in_{i1}}^2+\frac{1}{n_{i1}^l}\right)=\widehat{C}_{i1} \text{ and } n_{i1s_i}\geq\widehat{C}_{i1s_i}=\widehat{C}_{i1}a_{is_i}, \text{ for all } s_i=1,2,..,S_i,
\end{eqnarray}
where $t_1=\left(\frac{\delta^*}{h_1-2z_{(\alpha/2K)}}\right)^2$ and $h_1$ satisfies the following equation:
$
(K-1)\int_{-\infty}^{\infty}\Phi\left(y-h_1\right)\phi(y)dy=\frac{\delta}{2^{b}}.
$
Here $\delta>0$, and $b$ is chosen in such a way that $h_1>2z_{(\alpha/2K)}$ and $z_{(\alpha/2K)}$ denotes the $100\left(1-(\alpha/2K)\right)$ percentile of the standard normal distribution $N(0,1)$.
 The term $1/n_{i1}^l$ is added to prevent the procedure from stopping too early. In the absence of this term, the stopping criterion in \eqref{3.21} may be satisfied prematurely when only a few observations have been collected because $V_{in_i}^2$ can take small values during the early stages of sampling. In addition, $V_{in_{i1}}^2 + 1/n_{i1}^l$ remains a consistent estimator of $\xi_i^2$.

If the stopping criterion is satisfied, the procedure is terminated and the pilot sample size is taken as the final sample size for this stage. Otherwise, we select $m' (\geq 1)$ additional PSUs from each stratum satisfying $n_{i1s_i}<\widehat{C}_{i1s_i}$. From each of these newly selected PSUs, $k$ SSUs are chosen at random. Based on the updated data, $\xi_i^2$ is re-estimated and  the stopping rule is checked again. This process is repeated in the first stage until the stopping criterion is met. The resulting final number of PSUs in stratum $s_i$ is given by $N_{i1s_i} = N_{i1} a_{is_i}$.
Based on final PSU size $N_{i1}$ for the first stage, the estimator of $g\left(\bm{\theta}_i\right)$ is given by $\widetilde{g}\left(\widehat{\bm{\theta}}_{in_0},\widehat{\bm{\theta}}_{i(N_{i1}-n_0)}\right)$.
Now, eliminate $j^{th}$ population if
$\min\limits_{i=1,2,..,K}UCB_{i1}<LCB_{j1},$
where
$LCB_{i1}=\widetilde{g}\left(\widehat{\bm{\theta}}_{in_0},\widehat{\bm{\theta}}_{i(N_{i1}-n_0)}\right)-z_{(\alpha/2K)}\sqrt{t_1}$
and
$UCB_{i1}=\widetilde{g}\left(\widehat{\bm{\theta}}_{in_0},\widehat{\bm{\theta}}_{i(N_{i1}-n_0)}\right)+z_{(\alpha/2K)}\sqrt{t_1}.$
Let $L_1$ denote the set of indices corresponding to the populations eliminated at the first stage. The procedure then continues to the next stage using the remaining populations. Let $L$ represent the set of indices of the populations retained after elimination.\\
\noindent\textit{Stage $r\geq2:$ } In the $r^{th}$ stage, sampling is continued only for those populations that were not eliminated in the $(r-1)^{th}$ stage. For each $a\in L$, where $L$ denotes the set of remaining populations, sampling proceeds according to a stopping rule. Specifically, let $N_{ar} (\leq H_a)$ is the smallest integer  $n_{ar} \left(\geq N_{a(r-1)}\right)$, for which
\begin{eqnarray}\label{3.19}
n_{ar}\geq\frac{1}{t_r}\left(V_{an_{ar}}^2+\frac{1}{n_{ar}^l}\right)=\widehat{C}_{ar} \text{ and } n_{ars_a}\geq\widehat{C}_{ars_a}=\widehat{C}_{ar}a_{as_a}, \text{ for all } s_a=1,2,..,S_a,
\end{eqnarray}
where $t_r=\left(\frac{\delta^*}{h_r-2z_{(\alpha/2K)}}\right)^2$ and $h_r$ satisfy the following equation
$
(K-1)\int_{-\infty}^{\infty}\Phi\left(y-h_r\right)\phi(y)dy=\frac{\delta}{2^{r+b}}.
$
Here $\delta>0$ and $b$ is the same as chosen in the first stage. 

If the above condition is satisfied, the procedure is terminated and the final sample size at $(r-1)^{th}$ stage is taken as the final sample size for $r^{th}$ stage. Otherwise, we select $m' (\geq 1)$ additional PSUs from each stratum for which $n_{ars_a}<\widehat{C}_{ars_a}$. From each newly selected PSU, $k$ SSUs are chosen randomly. Using the updated sample, $\xi_a^2$ is re-estimated and check our stopping rule. This process is repeated at $r^{th}$ stage until the stopping criterion is met. The resulting final PSU size in stratum $s_a$ is given by $N_{ars_a} = N_{ar}a_{as_a}$.
Based on final PSU size $N_{ar}$ for the $r^{th}$ stage, the estimator of $g\left(\bm{\theta}_a\right)$ is given by $\widetilde{g}\left(\widehat{\bm{\theta}}_{an_0},\widehat{\bm{\theta}}_{a(N_{ar}-n_0)}\right)$.
Now, based on these estimators, elimination criteria is as follows: For $p\in L$, eliminate $p^{th}$ population at $r^{th}$ stage if
$\min\limits_{a\in L} UCB_{ar}<LCB_{pr},$
where $LCB_{ar}=\widetilde{g}\left(\widehat{\bm{\theta}}_{an_0},\widehat{\bm{\theta}}_{a(N_{ar}-n_0)}\right)-z_{(\alpha/2K)}\sqrt{t_r} \text{ and }
UCB_{ar}=\widetilde{g}\left(\widehat{\bm{\theta}}_{an_0},\widehat{\bm{\theta}}_{a(N_{ar}-n_0)}\right)+z_{(\alpha/2K)}\sqrt{t_r}.$
Let $L_r\subset L$ represent the set of indices of populations eliminated at the $r^{th}$ stage. Update the set of remaining populations by replacing $L$ with $L\backslash L_r$ and increase the stage index to $r+1$. The procedure is continued until only one population remains, that is, until $|L|=1$. The population remaining at the end of the procedure is declared as the extreme population by the algorithm.
\section{Theoretical and Simulation Results}\label{3.S4}
\noindent This section presents the theoretical and simulation results corresponding to the algorithms introduced in Section \ref{3.S3}.
\subsection{Theoretical Results}\label{3.S4.1}
The primary objective of these algorithms is to identify the true extreme population on the basis of $g(\bm\theta_i)$ among the given $K$ populations, provided $\underaccent{\tilde}{g} \in \Omega(\delta^*)$. The following theorems establish that the proposed procedures select the true extreme population with high probability. In particular, Theorem \ref{3.TH2} shows that the online algorithm as well as MAB-based algorithm asymptotically identifies the extreme population with high probability. To enhance the readability of the paper, all proofs of the stated theorem are deferred to the Appendix. 
\begin{theorem}\label{3.TH2}
Under the assumptions \ref{3.AN1}-\ref{3.AN15} stated in the Appendix, suppose that for $i\in\{1,2,\ldots,K\}$, the estimator $g\left(\bm{\widehat{\theta}}_{in_i}\right)$ satisfies the uniform continuity in probability condition.
Then the following results holds
\begin{enumerate}
\item [$(i)$] For $t>0$, the online algorithm described in Subsection \ref{3.S3.1} stops with probability one, that is, $P(N_i<\infty)=1$, for all $i=1,2,\ldots,K$.
\item [$(ii)$]
For the online algorithm described in Subsection \ref{3.S3.1}, we have
\[\liminf_{\delta^* \to 0}P(\text{Correct Selection})\geq 1-\delta, \text{ for all }\underaccent{\tilde}{g} \in \Omega(\delta^*),
\]
where $0<\delta<1$.
\item [$(iii)$]
The MAB-based algorithm described in Subsection \ref{3.S3.2} stops with probability at least $1-\alpha$, where $0<\alpha<1$.
\item [$(iv)$]
Suppose that for all $j\in\{2,3,\ldots,K\}$,
\[\sum_{n=n_0+2}^{\infty}\sum_{m=n_0+2}^{\infty}P\left(\widetilde{g}\left(\widehat{\bm{\theta}}_{(j)n_0},\widehat{\bm{\theta}}_{(j)(m-n_0)}\right)<\widetilde{g}\left(\widehat{\bm{\theta}}_{(1)n_0},\widehat{\bm{\theta}}_{(1)(n-n_0)}\right)
\right)<\infty,
\]
where $\widetilde{g}\left(\widehat{\bm{\theta}}_{(i)n_0},\widehat{\bm{\theta}}_{(i)(n-n_0)}\right)$ represent the estimator of measure $g(\bm\theta_{[i]})$.
Then, for MAB algorithm described in Subsection \ref{3.S3.2}, we have
$\liminf_{\delta^* \to 0}P(Correct\; Selection)\geq1-\delta, \text{ for all }\underaccent{\tilde}{g} \in \Omega(\delta^*), $
where $0<\delta<1$.
\end{enumerate}
\end{theorem}
\noindent\textbf{Proof of Theorem \ref{3.TH2} (i):}
Let $t>0$ be given. Then, for every $i=1,2,\ldots,K$, from equation \eqref{3.20}, we have
$
P(N_i=\infty)=\lim_{n_i\rightarrow\infty}P(N_i>n_i)
=\lim_{n_i\rightarrow\infty}P\left(n_i<\frac{1}{t}\left(V_{in_{i}}^2+\frac{1}{n_{i}^l}\right)\right).
$
Since $V_{in_{i}}^2$ is consistent estimator of $\xi_i^2$, we get
$\lim_{n_i\rightarrow\infty}P\left(n_i<\frac{1}{t}\left(V_{in_{i}}^2+\frac{1}{n_{i}^l}\right)\right)=0.$
This implies $P(N_i=\infty)=0.$ Hence, $P(N_i<\infty)=1$, for all $i=1,2,\ldots,K$.\\
\noindent\textbf{Proof of Theorem \ref{3.TH2} (ii):}
The online algorithm proposed in Section \ref{3.S3} correctly identifies the best population whenever, for every $i=2,3,\ldots,K$, $\widetilde{g}\left(\widehat{\bm{\theta}}_{(1)n_0},\widehat{\bm{\theta}}_{(1)(N_{(1)}-n_0)}\right)<\widetilde{g}\left(\widehat{\bm{\theta}}_{(i)n_0},\widehat{\bm{\theta}}_{(i)(N_{(i)}-n_0)}\right).$
Thus, the probability of correct selection is given by
\begin{eqnarray*}
P(CS)&=&P\left(\widetilde{g}\left(\widehat{\bm{\theta}}_{(1)n_0},\widehat{\bm{\theta}}_{(1)(N_{(1)}-n_0)}\right)<\widetilde{g}\left(\widehat{\bm{\theta}}_{(i)n_0},\widehat{\bm{\theta}}_{(i)(N_{(i)}-n_0)}\right), i=2,..,K\right)\\&=&P\bigg(\frac{\widetilde{g}\left(\widehat{\bm{\theta}}_{(1)n_0},\widehat{\bm{\theta}}_{(1)(N_{(1)}-n_0)}\right)-g\left(\bm{\theta}_{[1]}\right)}{\sqrt{t}}<\frac{\widetilde{g}\left(\widehat{\bm{\theta}}_{(i)n_0},\widehat{\bm{\theta}}_{(i)(N_{(i)}-n_0)}\right)-g\left(\bm{\theta}_{[i]}\right)}{\sqrt{t}}\\&&\hspace{6.5cm}+\frac{g\left(\bm{\theta}_{[i]}\right)-g\left(\bm{\theta}_{[1]}\right)}{\sqrt{t}};i=2,..,K\bigg)\\&\geq& P\bigg(\frac{\widetilde{g}\left(\widehat{\bm{\theta}}_{(1)n_0},\widehat{\bm{\theta}}_{(1)(N_{(1)}-n_0)}\right)-g\left(\bm{\theta}_{[1]}\right)}{\sqrt{t}}<\frac{\widetilde{g}\left(\widehat{\bm{\theta}}_{(i)n_0},\widehat{\bm{\theta}}_{(i)(N_{(i)}-n_0)}\right)-g\left(\bm{\theta}_{[i]}\right)}{\sqrt{t}}+\frac{\delta^*}{\sqrt{t}};\;i=2,..,K\bigg).
\end{eqnarray*}
Consequently, we get
\begin{eqnarray*}
\text{Inf } P(CS)&\geq& P\bigg(\frac{\widetilde{g}\left(\widehat{\bm{\theta}}_{(1)n_0},\widehat{\bm{\theta}}_{(1)(N_{(1)}-n_0)}\right)-g\left(\bm{\theta}_{[1]}\right)}{\sqrt{t}}<\frac{\widetilde{g}\left(\widehat{\bm{\theta}}_{(i)n_0},\widehat{\bm{\theta}}_{(i)(N_{(i)}-n_0)}\right)-g\left(\bm{\theta}_{[i]}\right)}{\sqrt{t}}\\&&\hspace{8.5cm}+\frac{\delta^*}{\sqrt{t}}; i=2,..,K\bigg).
\end{eqnarray*}
Now, on taking limits both sides, we get
\begin{eqnarray*}
\liminf_{\delta^*\rightarrow0}P(CS)&\geq&\lim_{\delta^*\rightarrow0}P\bigg(\frac{\widetilde{g}\left(\widehat{\bm{\theta}}_{(1)n_0},\widehat{\bm{\theta}}_{(1)(N_{(1)}-n_0)}\right)-g\left(\bm{\theta}_{[1]}\right)}{\sqrt{t}}<\\&&\hspace{2.5cm}\frac{\widetilde{g}\left(\widehat{\bm{\theta}}_{(i)n_0},\widehat{\bm{\theta}}_{(i)(N_{(i)}-n_0)}\right)-g\left(\bm{\theta}_{[i]}\right)}{\sqrt{t}}+\frac{\delta^*}{\sqrt{t}}; i=2,..,K\bigg).
\end{eqnarray*}
Consequently, for $t=\left(\frac{\delta^*}{h}\right)^2$, we get
\begin{eqnarray*}
\liminf_{\delta^*\rightarrow0}P(CS)&\geq&\lim_{\delta^*\rightarrow0}P\bigg(\frac{\widetilde{g}\left(\widehat{\bm{\theta}}_{(1)n_0},\widehat{\bm{\theta}}_{(1)(N_{(1)}-n_0)}\right)-g\left(\bm{\theta}_{[1]}\right)}{\sqrt{t}}<\\&&\hspace{2.5cm}\frac{\widetilde{g}\left(\widehat{\bm{\theta}}_{(i)n_0},\widehat{\bm{\theta}}_{(i)(N_{(i)}-n_0)}\right)-g\left(\bm{\theta}_{[i]}\right)}{\sqrt{t}}+h;i=2,..,K\bigg).
\end{eqnarray*}
Since populations are independent to each other, by using Lemma \ref{3.L5}, we get
\begin{eqnarray*}
\liminf_{\delta^* \to 0}P(CS)\geq\int_{-\infty}^{\infty}\left(1-\Phi\left(y-h\right)\right)^{K-1}\phi(y)dy=1-\delta.
\end{eqnarray*}
Hence the result.\\
\noindent\textbf{Proof of Theorem \ref{3.TH2} (iii):}
From Lemma \ref{3.L6}, there exists some $r_1$ such that
\begin{eqnarray}\label{3.17}
\widetilde{g}\left(\widehat{\bm{\theta}}_{in_0},\widehat{\bm{\theta}}_{i(N_{ir}-n_0)}\right)-z_{(\alpha/2K)}\sqrt{t_r}<g\left(\bm{\theta}_i\right)<\widetilde{g}\left(\widehat{\bm{\theta}}_{in_0},\widehat{\bm{\theta}}_{i(N_{ir}-n_0)}\right)+z_{(\alpha/2K)}\sqrt{t_r}
\end{eqnarray}
holds with probability at least $1-\alpha$, for every $r>r_1$ and for all $i=1,2,\ldots,K$.
Now, suppose that for some $r>r_1$ and for some $i\neq j \in\{1,2,\ldots,K\}$,
$g\left(\bm{\theta}_i\right)-g\left(\bm{\theta}_j\right)\geq4z_{(\alpha/2K)}\sqrt{t_r}.
$
Then, using equation \eqref{3.17},
\begin{eqnarray*}
4z_{(\alpha/2K)}\sqrt{t_r}&\leq& g\left(\bm{\theta}_i\right)-g\left(\bm{\theta}_j\right)\\&<&\widetilde{g}\left(\widehat{\bm{\theta}}_{in_0},\widehat{\bm{\theta}}_{i(N_{ir}-n_0)}\right)+z_{(\alpha/2K)}\sqrt{t_r}-\widetilde{g}\left(\widehat{\bm{\theta}}_{jn_0},\widehat{\bm{\theta}}_{j(N_{jr}-n_0)}\right)+z_{(\alpha/2K)}\sqrt{t_r}
\end{eqnarray*}
holds with probability at least $1-\alpha$. Consequently, we have
\[\widetilde{g}\left(\widehat{\bm{\theta}}_{jn_0},\widehat{\bm{\theta}}_{j(N_{jr}-n_0)}\right)+z_{(\alpha/2K)}\sqrt{t_r}<\widetilde{g}\left(\widehat{\bm{\theta}}_{in_0},\widehat{\bm{\theta}}_{i(N_{ir}-n_0)}\right)-z_{(\alpha/2K)}\sqrt{t_r}
\]
with probability at least $1-\alpha$. Hence, the $i^{th}$ population is eliminated at some stage $r>r_1$ with probability at least $1-\alpha$. Now suppose that the algorithm never terminates. Then there exists $i,j\in\{1,2,\ldots,K\}$ such that
\[\widetilde{g}\left(\widehat{\bm{\theta}}_{jn_0},\widehat{\bm{\theta}}_{j(N_{jr}-n_0)}\right)+z_{(\alpha/2K)}\sqrt{t_r}>\widetilde{g}\left(\widehat{\bm{\theta}}_{in_0},\widehat{\bm{\theta}}_{i(N_{ir}-n_0)}\right)-z_{(\alpha/2K)}\sqrt{t_r}\]
holds for all $r>r_1$. Thus, we have
$
0<g\left(\bm{\theta}_i\right)-g\left(\bm{\theta}_j\right)\leq4z_{(\alpha/2K)}\sqrt{t_r}
$
with probability at least $1-\alpha$, for all $r>r_1$. Therefore,
$g\left(\bm{\theta}_i\right)-g\left(\bm{\theta}_j\right)\rightarrow0 \; \text{as } r\rightarrow\infty.
$
This contradicts the fact that $g\left(\bm{\theta}_i\right)\neq g\left(\bm{\theta}_j\right)$. Hence, the algorithm terminates with probability at least $1-\alpha$.

\noindent\textbf{Proof of Theorem \ref{3.TH2} (iv):}
The MAB-based algorithm proposed in Section \ref{3.S3} correctly identifies the best population if the best population is not eliminated at any stage. Thus, the best population remains active throughout all stages if
$
LCB_{(1)r}<\min_{j=2,\ldots,K} UCB_{(j)r},
$
for every $r\in\mathbb{N}$. Thus, the probability of correct selection is given by
\begin{eqnarray*}
P(CS)&\geq& P\left(LCB_{(1)r}<min_{j=2,..,K}UCB_{(j)r},\text{ for all }r\in\mathbb{N}\right)\\
&=&P\bigg(\widetilde{g}\left(\widehat{\bm{\theta}}_{(1)n_0},\widehat{\bm{\theta}}_{(1)(N_{(1)r}-n_0)}\right)-z_{\frac{\alpha}{2K}}\sqrt{t_r}<\widetilde{g}\left(\widehat{\bm{\theta}}_{(j)n_0},\widehat{\bm{\theta}}_{(j)(N_{(j)r}-n_0)}\right)+z_{\frac{\alpha}{2K}}\sqrt{t_r};j=2,..,K,r\in\mathbb{N}\bigg)
\end{eqnarray*}
Using Boole's inequality, we get
\begin{eqnarray*}
P(CS)&\geq&1-\sum_{r=1}^{\infty}\sum_{j=2}^{K}P\bigg\{\frac{\widetilde{g}\left(\widehat{\bm{\theta}}_{(j)n_0},\widehat{\bm{\theta}}_{(j)(N_{(j)r}-n_0)}\right)-g\left(\bm{\theta}_{[j]}\right)}{\sqrt{t_r}}<\\&&\hspace{2.5cm}\frac{\widetilde{g}\left(\widehat{\bm{\theta}}_{(1)n_0},\widehat{\bm{\theta}}_{(1)(N_{(1)r}-n_0)}\right)-g\left(\bm{\theta}_{[1]}\right)}{\sqrt{t_r}}-\frac{g\left(\bm{\theta}_{[j]}\right)-g\left(\bm{\theta}_{[1]}\right)}{\sqrt{t_r}}-2z_{\frac{\alpha}{2K}}\bigg\}
\\&\geq&1-\sum_{r=1}^{\infty}\sum_{j=2}^{K}P\bigg\{\frac{\widetilde{g}\left(\widehat{\bm{\theta}}_{(j)n_0},\widehat{\bm{\theta}}_{(j)(N_{(j)r}-n_0)}\right)-g\left(\bm{\theta}_{[j]}\right)}{\sqrt{t_r}}<\\&&\hspace{4cm}\frac{\widetilde{g}\left(\widehat{\bm{\theta}}_{(1)n_0},\widehat{\bm{\theta}}_{(1)(N_{(1)r}-n_0)}\right)-g\left(\bm{\theta}_{[1]}\right)}{\sqrt{t_r}}-\frac{\delta^*}{\sqrt{t_r}}-2z_{\frac{\alpha}{2K}}\bigg\}.
\end{eqnarray*}
Consequently, for $t_r=\left(\frac{\delta^*}{h_r-2z_{(\alpha/2K)}}\right)^2$, we get\\
\resizebox{\textwidth}{!}{%
$\text{Inf} P(CS)\geq1-\text{sup}\sum_{r=1}^{\infty}\sum_{j=2}^{K}P\bigg\{\frac{\widetilde{g}\left(\widehat{\bm{\theta}}_{(j)n_0},\widehat{\bm{\theta}}_{(j)(N_{(j)r}-n_0)}\right)-g\left(\bm{\theta}_{[j]}\right)}{\sqrt{t_r}}<\frac{\widetilde{g}\left(\widehat{\bm{\theta}}_{(1)n_0},\widehat{\bm{\theta}}_{(1)(N_{(1)r}-n_0)}\right)-g\left(\bm{\theta}_{[1]}\right)}{\sqrt{t_r}}-h_r\bigg\}.$%
}
Now, on taking limits both sides, we get\\
\begin{eqnarray*}
\liminf_{\delta^*\rightarrow0}P(CS)&\geq& 1- \limsup_{\delta^*\rightarrow0}\sum_{r=1}^{\infty}\sum_{j=2}^{K}P\Bigg\{\frac{\widetilde{g}\left(\widehat{\bm{\theta}}_{(j)n_0},\widehat{\bm{\theta}}_{(j)(N_{(j)r}-n_0)}\right)-g\left(\bm{\theta}_{[j]}\right)}{\sqrt{t_r}}<\\&&\hspace{4.5cm}\frac{\widetilde{g}\left(\widehat{\bm{\theta}}_{(1)n_0},\widehat{\bm{\theta}}_{(1)(N_{(1)r}-n_0)}\right)-g\left(\bm{\theta}_{[1]}\right)}{\sqrt{t_r}}-h_r\Bigg\}.
\end{eqnarray*}
Consequently, by using Lemma \ref{3.L9}, we get
$\liminf_{\delta^*\rightarrow0}P(CS)\geq 1-\delta.$ 
Hence the result.\\
\begin{remark}
 The statements and corresponding proofs of Lemmas \ref{3.L1}–\ref{3.L9} are provided in Appendix \ref{3.8.2}.
\end{remark}
\begin{corollary}\label{3.C2}
Under the conditions of Theorem \ref{3.TH2}, both the online algorithm and MAB-based algorithm, when applied to the Gini index and the mean of the log-transformed TMB scores discussed in Subsection \ref{3.S2.3}, satisfy
\[\liminf_{\delta^* \to 0}P(Correct\; Selection)\geq 1-\delta,\]
where $0<\delta<1$.
\end{corollary}
\noindent\textbf{Proof of Corollary \ref{3.C2}:}
Using Lemma 2 and Theorem 2 of \cite{shivam2026asymptotic},  $\widehat{\mu}_{TMB_{in_i}}$  and $\widehat{G}_{in_i}$ satisfy uniform continuity in probability condition. Consequently, from Lemma \ref{3.L5}, we get
\[\frac{\widetilde{G}_{iN_i}-G_i}{\sqrt{t}}\rightarrow N\left(0,1\right),\text{ and }\;\frac{\widetilde{\mu}_{TMB_{iN_i}}-\mu_{TMB_i}}{\sqrt{t}}\rightarrow N\left(0,1\right),\]
where $\widetilde{G}_{iN_i}=c_i\widehat{G}_{in_0}+(1-c_i)\widehat{G}_{i(N_i-n_0)}$ and $\widetilde{\mu}_{TMB_{iN_i}}=c_i\widehat{\mu}_{TMB_{in_0}}+(1-c_i)\widehat{\mu}_{TMB_{i(N_i-n_0)}}.$
Hence, for both the algorithms, from Theorem \ref{3.TH2}, we get $\liminf_{\delta^* \to 0}P(Correct\; Selection)\geq 1-\delta.$
Hence, the result.
\subsection{Simulation Results}\label{3.S4.2}
In this subsection, we assess the performance of the proposed algorithms using a Monte Carlo simulation study. The primary aim is to investigate how effectively the algorithms identify the extreme population. To illustrate the applicability of the proposed framework, the simulation experiments are conducted using the Gini index introduced in Subsection \ref{3.S2.3.2}.
\subsubsection{Pseudo-population}\label{3.S4.2.1}
We generate three pseudo-populations. For each $i=1,2,3$, the $i^{th}$ population is divided into two strata indexed by $s_i=1,2$, where each stratum contains $H_{is_i}$ clusters. Further, within every cluster, households are divided into two sub-strata. The variable of interest, such as household income or expenditure, is denoted by $x_{s_ic_{is_i}b_{ic_{is_i}}h}$.

Household incomes are generated from distributions commonly used in the income inequality literature, namely Gamma $(2.231,0.84)$, Gamma $(3.231,0.64)$, Gamma $(2.551,0.79)$, Pareto $(20000,5)$, Pareto $(20000,4.1)$, Pareto $(20000,3.5)$, lognormal $(2.597,0.496)$, lognormal $(3.497,0.56)$ and lognormal $(3.997,0.64)$, with parameter values adopted from \cite{INCOME}. The three pseudo-populations are chosen in such a manner that the true Gini index values satisfy the preference zone condition given in \eqref{3.preferencezone} and also no two Gini index value coincide with each other. 

We next assess the performance of the online algorithm and the MAB-based algorithm in terms of their ability to correctly identify the true extreme population, namely the population having the minimum value of Gini index. Throughout the simulation study, the desired probability of correct selection is fixed at $95\%$. For different combinations of income distributions, we conduct $10000$ Monte Carlo replication. Each pseudo-population consists of two strata, each containing $H_{is_i}=5000$ clusters. Pilot cluster sizes are computed using equation \eqref{3.13}. From every chosen cluster, $k=2$ households are drawn within each sub-stratum and the corresponding sampling weights are computed in accordance with the survey design. Both procedures are implemented to determine the true extreme population.

\subsubsection{Simulation results for Online Algorithm}\label{3.S4.2.2}
We evaluate the performance of online algorithm discussed in Section \ref{3.S3} under the different combinations of distributions. For each $i$, using the variance estimator \eqref{3.VARIANCEGINI} and the stopping rule \eqref{3.20}, we determine the final cluster sample size $N_i$. Consequently, we compute the Gini index estimator given in equation \eqref{3.gGINI}. The population corresponding to minimum Gini index estimator is then considered as extreme population. The procedure is repeated $10000$ times. The empirical results are reported in Table \ref{3.TA1}.

The first, second and third columns of Table \ref{3.TA1} list the distribution of monthly income. The fourth, fifth and sixth columns give the true value of Gini index corresponding to population provided in column first, second and third, respectively. The seventh column reports the pre-specified value of $\delta^*$. The eighth, ninth and tenth columns report the mean final number of clusters $(\bar{N}_1, \bar{N}_2, \bar{N}_3)$ together with their respective standard deviations $(sd(N_1), sd(N_2), sd(N_3))$. The last column presents the proportion of simulation runs in which the algorithm correctly identifies the true extreme population, denoted by $\widehat{P}(CS)$, together with its corresponding standard error $se(\widehat{P}(CS))$, where $CS$ stands for Correct Selection.

The results in Table \ref{3.TA1} confirm that the online algorithm identifies the true extreme population with probability close to $100(1-\alpha)\%$.
\begin{table}
\caption{Results for online algorithm for Gini index}
\centering
\resizebox{\textwidth}{!}{%
\begin{tabular}{ccccccccccccc}
		\toprule
$1^{st}$&$2^{nd}$&$3^{rd}$&$G_{X_1}$&$G_{X_2}$&$G_{X_3}$&$\delta^*$&$\overline{N}_1$&$\overline{N}_2$&$\overline{N}_3$&$\widehat{P}\left(CS\right)$\\
 Population&Population&Population&&&&&$sd(N_1)$&$sd(N_2)$&$sd(N_3)$&$se(\widehat{P}\left(CS\right))$\\
    \midrule
     Gamma&Gamma&Gamma&$0.3573$&$0.3020$&$0.3364$&$0.0344$&$88.3974$&$68.4088$&$80.8716$&$0.9620$\\
     $(2.231,0.84)$&$(3.231,0.64)$&$(2.551,0.79)$&&&&&$(18.7700)$&$(15.3483)$&$(17.5700)$&$(0.0019)$\\
     Pareto&Pareto&Pareto&$0.1111$&$0.1389$&$0.1667$&$0.0277$&$56.9770$&$90.4924$&$144.8058$&$0.9390$\\
      $(20000,5)$&$(20000,4.1)$&$(20000,3.5)$&&&&&$(27.0043)$&$(49.9908)$&$(82.1658)$&$(0.0024)$\\
     lognormal&lognormal&lognormal&$0.2742$&$0.3079$&$0.3491$&$0.0226$&$69.6524$&$87.9344$&$115.5250$&$0.9612$\\
      $(2.597,0.496)$&$(3.497,0.56)$&$(3.997,0.64)$&&&&&$(21.1203)$&$(27.4380)$&$(35.6057)$&$(0.0019)$\\
     lognormal&lognormal&Pareto&$0.2742$&$0.3491$&$0.1667$&$0.1075$&$12.3742$&$14.6444$&$13.6054$&$0.9806$\\
      $(2.597,0.496)$&$(3.997,0.64)$&$(20000,3.5)$&&&&&$(2.6841)$&$(4.9861)$&$(6.7657)$&$(0.0014)$\\
     Pareto&Pareto&Gamma&$0.1389$&$0.1667$&$0.3020$&$0.0277$&$90.5324$&$143.6638$&$104.7198$&$0.9437$\\
      $(20000,4.1)$&$(20000,3.5)$&$(3.231,0.64)$&&&&&$(49.4332)$&$(80.6754)$&$(20.3099)$&$(0.0023)$\\
     Gamma&Gamma&lognormal&$0.3573$&$0.3364$&$0.3079$&$0.0285$&$129.0468$&$117.7582$&$125.4224$&$0.9711$\\
      $(2.231,0.84)$&$(2.551,0.79)$&$(3.497,0.56)$&&&&&$(23.1882)$&$(21.6027)$&$(34.3409)$&$(0.0017)$\\
		\bottomrule
\end{tabular}%
}
\label{3.TA1}
\end{table}
\subsubsection{Simulation results for MAB-based Algorithm}\label{3.S4.2.3}
We now evaluate the performance of MAB-based algorithm discussed in Section \ref{3.S3} under different combinations of distributions.  For each $i$, the final cluster sample size $N_i$ and the total number of stages $R$ are determined using the variance estimator given in \eqref{3.VARIANCEGINI}, together with the stopping rules defined in \eqref{3.21} and \eqref{3.19}. In the algorithm, the procedure continues until only one population remains and all other populations are eliminated. The population which is not eliminated is then declared as the extreme population. The procedure is repeated $10000$ times. The empirical results summarized in Table \ref{3.TA2}.

The first, second and third columns of Table \ref{3.TA2} list the distribution of monthly income. The fourth, fifth and sixth columns give the true value of Gini index corresponding to population provided in column first, second and third, respectively. The seventh column reports the pre-specified value of $\delta^*$. The eighth column determines the average number of stages required by procedure, denoted by $(\bar{R})$, along with the standard deviation $(sd(R))$. The ninth, tenth and eleventh columns report the mean final numbers of clusters $(\bar{N}_1, \bar{N}_2, \bar{N}_3)$ together with their associated standard deviations $(sd(N_1), sd(N_2), sd(N_3))$. The last column presents the proportion of simulation runs in which the algorithm correctly identifies the true extreme population, denoted by $\widehat{P}(CS)$, together with its corresponding standard error $se(\widehat{P}(CS))$.

The results in Table \ref{3.TA2} confirms that the MAB-based algorithm identifies the true extreme population with probability close to $1$ and final cluster size is higher as compared to online algorithm.
\begin{table}
\caption{Results for MAB-based algorithm for Gini index}
\centering
\resizebox{\textwidth}{!}{%
\begin{tabular}{cccccccccccccc}
		\toprule
$1^{st}$&$2^{nd}$&$3^{rd}$&$G_{X_1}$&$G_{X_2}$&$G_{X_3}$&$\delta^*$&$\overline{R}$&$\overline{N}_1$&$\overline{N}_2$&$\overline{N}_3$&$\widehat{P}\left(CS\right)$\\
 Population&Population&Population&&&&&$sd(R)$&$sd(N_1)$&$sd(N_2)$&$sd(N_3)$&$se(\widehat{P}\left(CS\right))$\\
   \midrule
     Gamma&Gamma&Gamma&$0.3573$&$0.3020$&$0.3364$&$0.0344$&$26.0965$&$124.3180$&$225.6172$&$265.5434$&$0.9973$\\
     $(2.231,0.84)$&$(3.231,0.64)$&$(2.551,0.79)$&&&&&$(9.5860)$&$(76.5420)$&$(131.8591)$&$(152.2147)$&$(0.0005)$\\
     Pareto&Pareto&Pareto&$0.1111$&$0.1389$&$0.1667$&$0.0277$&$27.0258$&$206.9114$&$367.9514$&$173.9214$&$0.9796$\\
      $(20000,5)$&$(20000,4.1)$&$(20000,3.5)$&&&&&$(10.4778)$&$(158.0205)$&$(212.3688)$&$(123.0412)$&$(0.0014)$\\
     lognormal&lognormal&lognormal&$0.2742$&$0.3079$&$0.3491$&$0.0226$&$25.8640$&$235.5540$&$304.1372$&$93.1380$&$0.9962$\\
      $(2.597,0.496)$&$(3.497,0.56)$&$(3.997,0.64)$&&&&&$(9.7717)$&$(152.0135)$&$(170.2087)$&$(54.2331)$&$(0.0006)$\\
     lognormal&lognormal&Pareto&$0.2742$&$0.3491$&$0.1667$&$0.1075$&$25.6914$&$25.1076$&$17.7334$&$30.9784$&$0.9979$\\
      $(2.597,0.496)$&$(3.997,0.64)$&$(20000,3.5)$&&&&&$(9.4475)$&$(12.8348)$&$(10.1929)$&$(26.5269)$&$(0.0004)$\\
     Pareto&Pareto&Gamma&$0.1389$&$0.1667$&$0.3020$&$0.0277$&$25.8929$&$343.6628$&$570.2156$&$15.5538$&$0.9919$\\
      $(20000,4.1)$&$(20000,3.5)$&$(3.231,0.64)$&&&&&$(10.5123)$&$(281.5621)$&$(331.1186)$&$(7.1591)$&$(0.0009)$\\
     Gamma&Gamma&lognormal&$0.3573$&$0.3364$&$0.3079$&$0.0285$&$25.0573$&$148.4886$&$368.2638$&$405.2138$&$0.9990$\\
      $(2.231,0.84)$&$(2.551,0.79)$&$(3.497,0.56)$&&&&&$(9.7975)$&$(87.7128)$&$(221.0137)$&$(267.5344)$&$(0.0003)$\\
		\bottomrule
\end{tabular}%
}
\label{3.TA2}
\end{table}
{\subsubsection{Comparison between MAB-based algorithm and online algorithm}
In this section, we compare the performances of the MAB-based algorithm and the online algorithm. For this purpose, we use the distributional combinations given in Tables \ref{3.TA1} and \ref{3.TA2}.  For each combination, the average final cluster sample size per population is computed across all simulation runs for both algorithms. These averages are then plotted. In the resulting scatter plots given in Figure \ref{3.ALLF}, red dots represent the MAB-based algorithm, whereas blue dots represent the online algorithm. 
From these plots, we see that, for each distributional combination, the average final cluster sample sizes per population obtained using the MAB-based algorithm display greater variability than those obtained using the online algorithm. From these figures, we can conclude that the online algorithm exhibits greater consistency in the final cluster sample size per population compared with the MAB-based algorithm. Further, Tables \ref{3.TA1} and \ref{3.TA2} show that the online algorithm generally requires fewer cluster samples on average than the MAB-based algorithm. Although the MAB-based algorithm requires a larger final cluster sample size, Figure \ref{3.F7} shows that it achieves greater accuracy, identifying the true extreme population with a probability close to $1$.
Based on the above comparison, the selection of an algorithm depends on the balance between sampling efficiency and selection accuracy. If the main objective is to attain the desired probability of correct selection with few cluster sample size, the online algorithm is the preferred choice. On the other hand, when accurate identification of the true extreme population is more important, the MAB-based algorithm may be preferred. It achieves a probability of correct selection close to one, but requires an excessively large cluster sample size. For instance, in clinical trials where the objective is to determine the most effective drug, the MAB-based algorithm is preferred, as an incorrect choice may lead to serious consequences. In contrast, in large-scale online advertising or marketing campaigns, the online algorithm is preferred because it can identify the better-performing option using limited data, thereby reducing experimentation cost while enabling faster decisions.}

\begin{figure}
    \begin{minipage}{0.47\textwidth}
    \centering
    \includegraphics[width=0.85\linewidth]{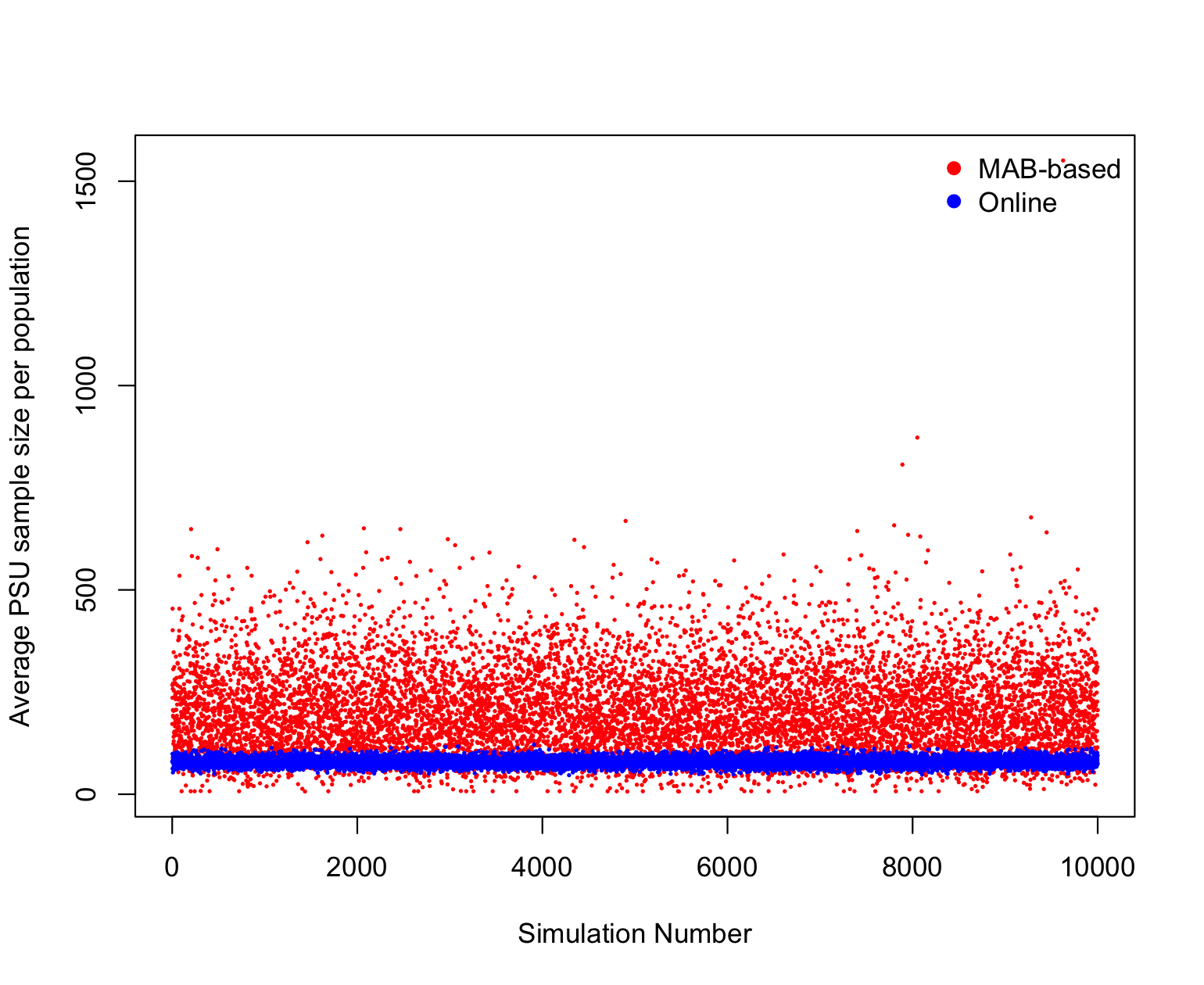}
    \subcaption{}
    \label{3.F1}
    \end{minipage}
    \hfill
    \begin{minipage}{0.47\textwidth}
    \centering
     \includegraphics[width=0.85\linewidth]{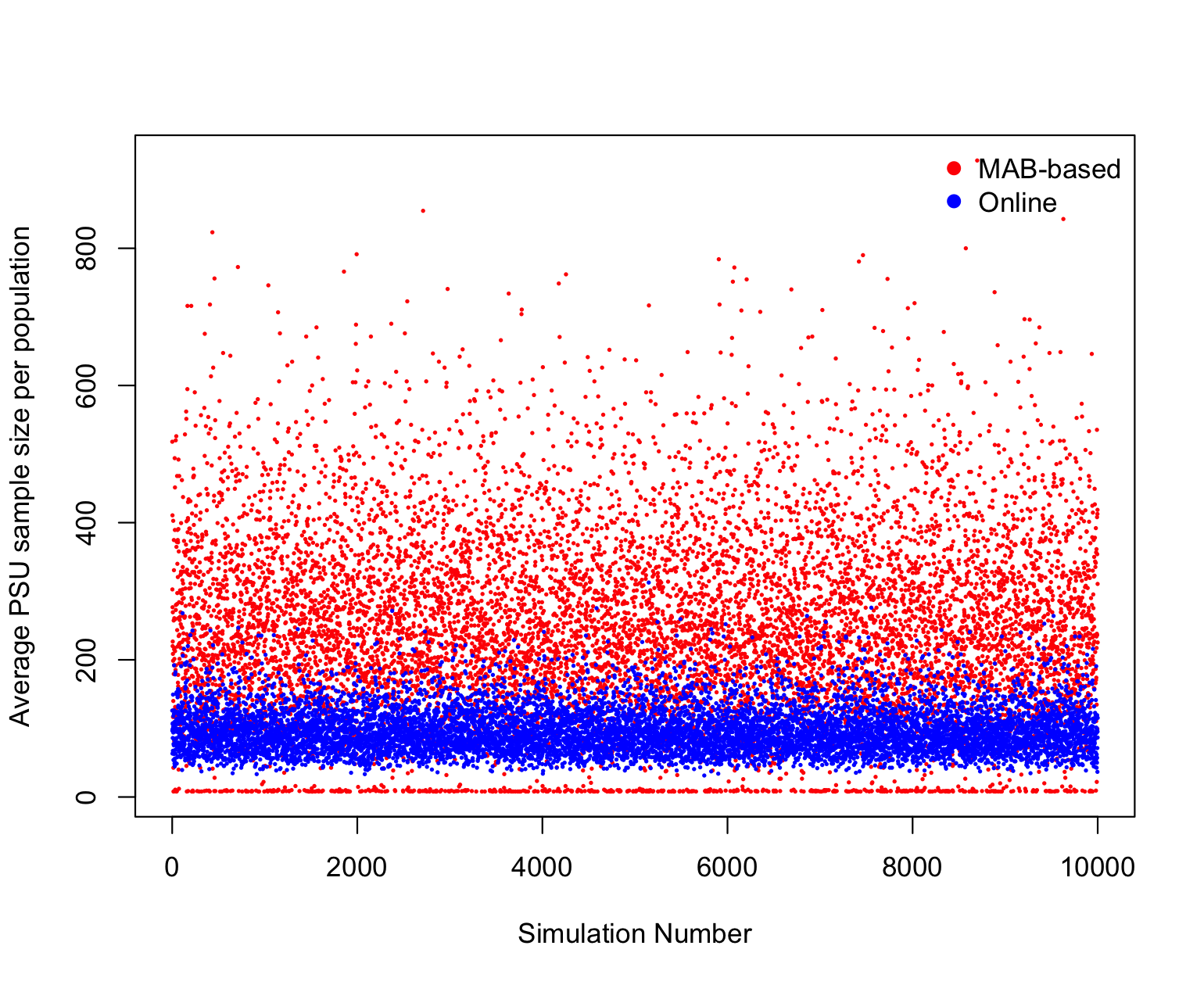}
     \subcaption{}
     \label{3.F2}
    \end{minipage}
    \begin{minipage}{0.47\textwidth}
    \centering
    \includegraphics[width=0.85\linewidth]{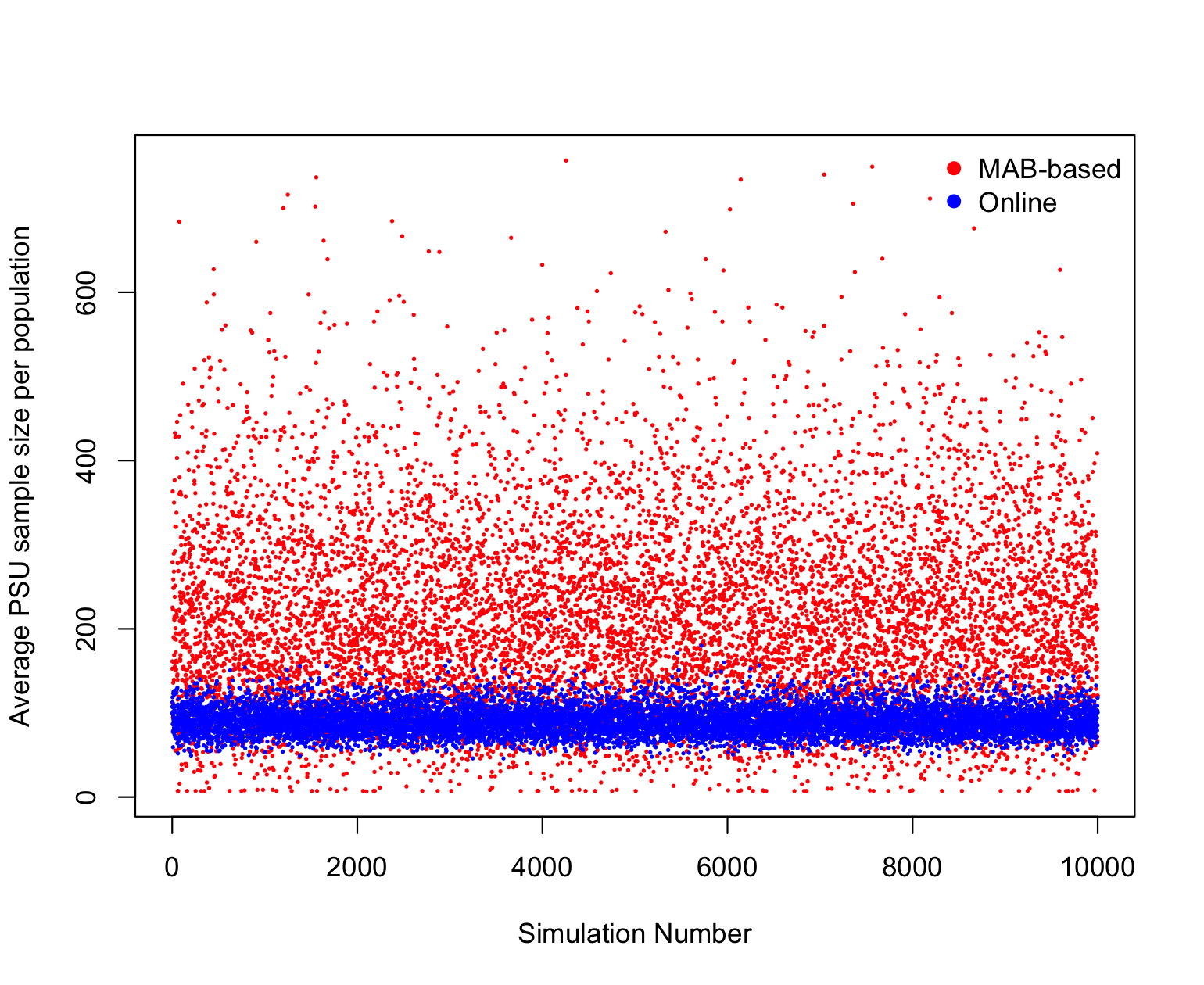}
    \subcaption{}
    \label{3.F3}
    \end{minipage}
    \hfill
    \begin{minipage}{0.47\textwidth}
    \centering
     \includegraphics[width=0.85\linewidth]{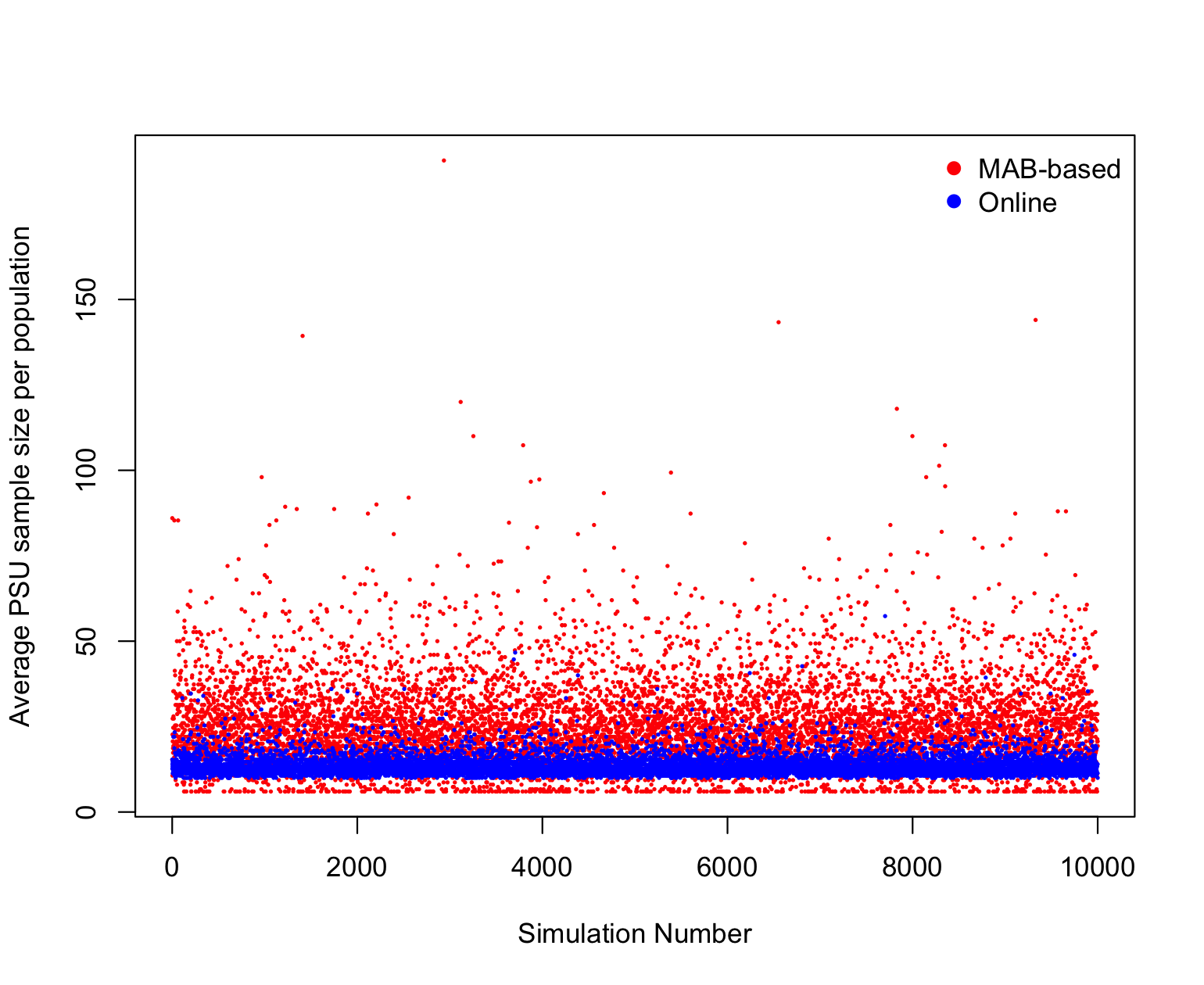}
     \subcaption{}
     \label{3.F4}
    \end{minipage}
\begin{minipage}{0.47\textwidth}
    \centering
    \includegraphics[width=0.85\linewidth]{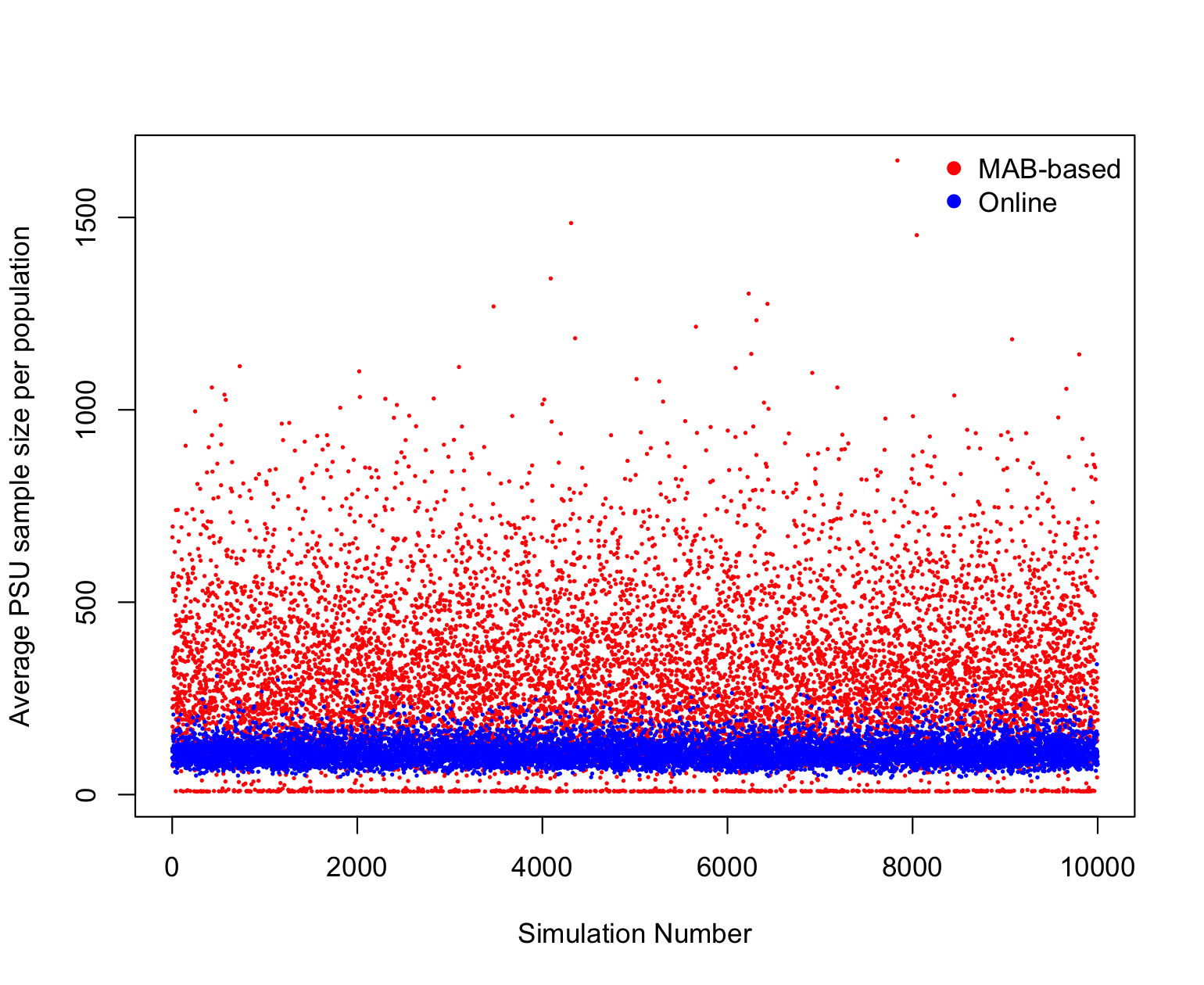}
    \subcaption{}
    \label{3.F5}
    \end{minipage}
    \hfill
    \begin{minipage}{0.47\textwidth}
    \centering
     \includegraphics[width=0.85\linewidth]{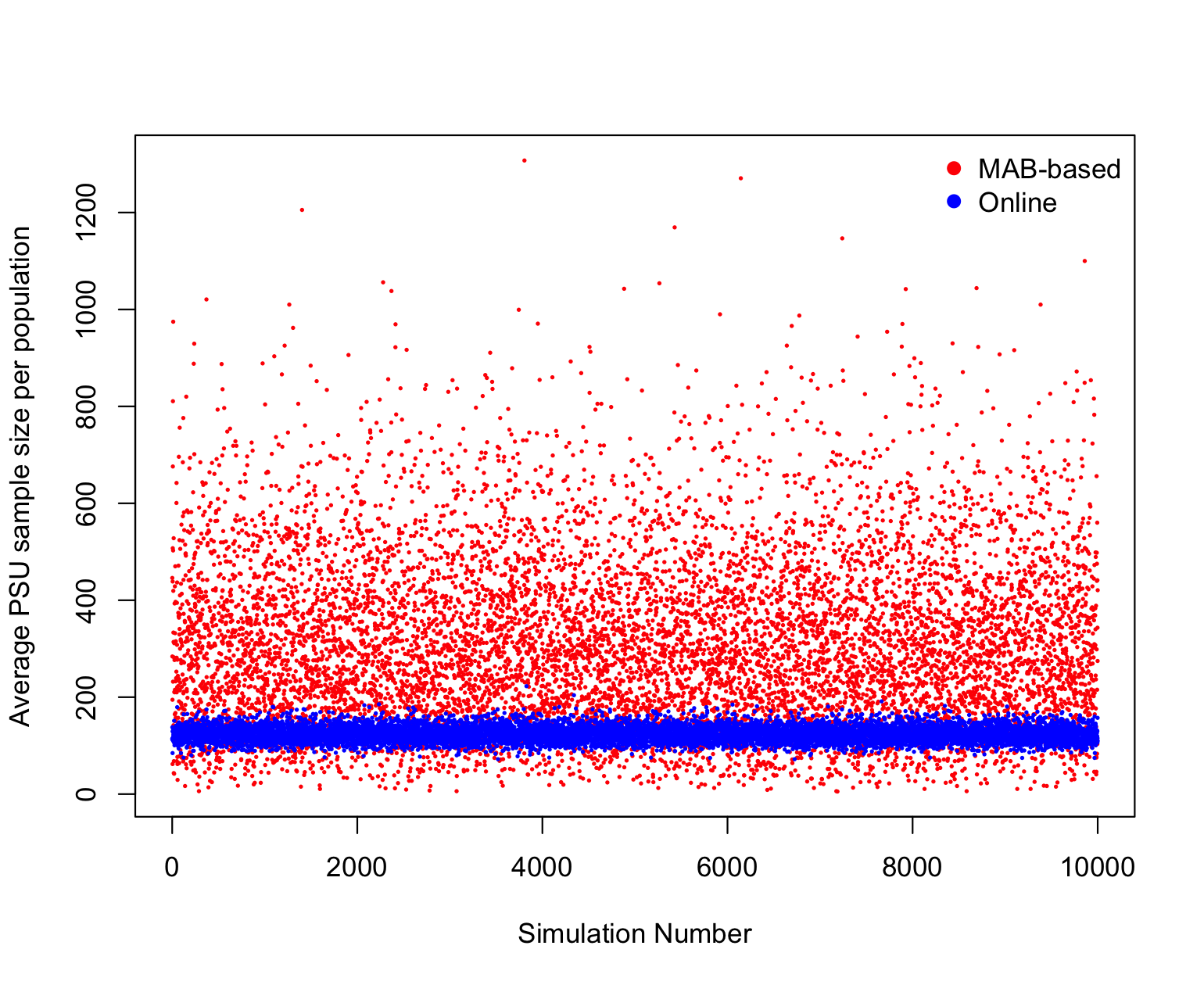}
     \subcaption{}
          \label{3.F6}
    \end{minipage}
    \caption{(a) All-Gamma population, (b) All-Pareto populations, (c) All-lognormal population, (d) Lognormal-Pareto combination, (e) Pareto-Gamma combination, (f) Gamma-Lognormal combination. }
    \label{3.ALLF}
    \end{figure}
    \begin{figure}
        \centering
        \includegraphics[width=0.6\linewidth]{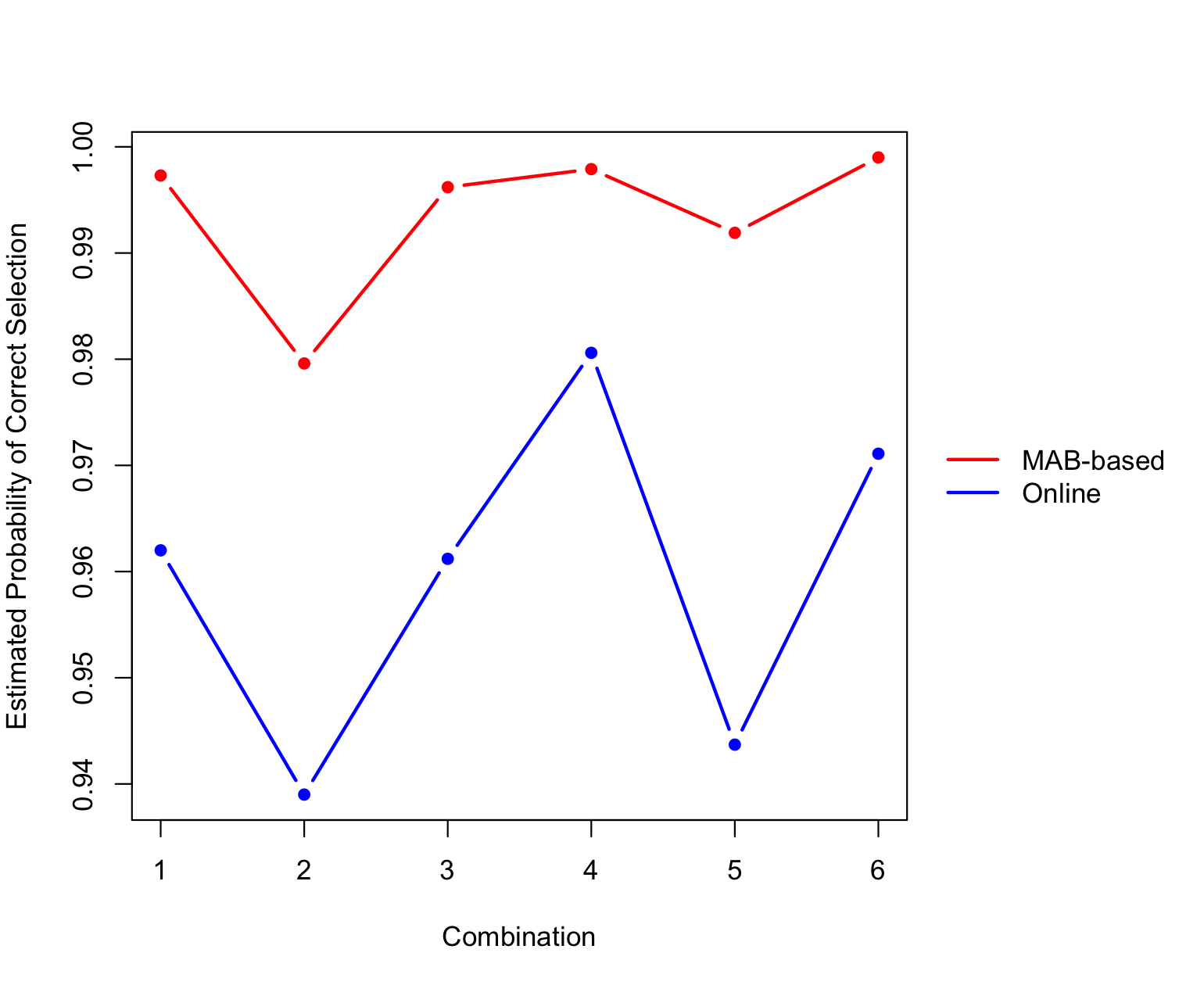}
        \caption{Estimated Probability of Correct Selection}
        \label{3.F7}
    \end{figure}
\section{Application}\label{3.S5}
This section illustrates the proposed methodology using the measure introduced in Subsection~\ref{3.S2.3.1}. The measure is the mean of $\log(TMB+1)$ scores applied to a real clinical sequencing cohort, the Memorial Sloan Kettering (MSK)-IMPACT $50000$\footnote{\url{https://www.cbioportal.org/study/clinicalData?id=msk_impact_50k_2026}}. The cohort is a pan-cancer clinical sequencing dataset containing $54,331$ tumors along with their matched normal samples, profiled through MSK-IMPACT. This dataset is used to identify the genetic ancestry which has minimum value of TMB-based measure defined in equation \eqref{3.TMB}.

The dataset contains several genetic ancestry categories, of which we focus on two namely ASJ\-EUR (Ashkenazi Jewish European ancestry), and nonASJ\-EUR (predominantly non-Ashkenazi European ancestry).
The dataset also includes ADMIX\_OTHER, AFR, EAS, SAS, and NAM groups denoting African, East Asian, South Asian, and Native or Admixed American ancestry, respectively, which are not used in the present analysis. We compare the two selected ancestry groups using the tumor mutational burden (TMB)-based measure defined in equation \eqref{3.TMB}.

After cleaning the dataset, we define strata based on cancer type described in Subsection~\ref{3.S2.3.1}. The individual patients serve as the ultimate sampling units (PSUs). For our study, observations with missing TMB scores are removed from the dataset. If multiple samples are available for the same patient within a stratum, they are replaced by the average TMB score across all of that patient's observations. As TMB scores are typically highly right-skewed and may contain zero values for certain patients, a logarithmic transformation of the form $\log(\text{TMB}+1)$ is applied to the TMB scores. Within each ancestry group, we treat the resulting data as a finite population and define strata according to the multistage sampling framework in Subsection \ref{3.S2.3.1} with patients being the primary sampling units. We considered the case in which strata are formed based on cancer type.


For cancer-type stratification, we restrict attention to five cancer types namely breast, colorectal, non-small cell lung, pancreatic, and prostate cancer, and include only the ASJ-EUR and nonASJ-EUR ancestry groups, since ADMIX$\_$OTHER has an insufficient number of observations in this subset. 
The number of observations in each population is provided in Table \ref{3.TA3}.
\begin{table}[h]
\caption{Number of observation for population under consideration}
\centering
\begin{tabular}{ccc}
\toprule
Stratification type&ASJ-EUR&nonASJ-EUR\\
\midrule
Cancer-type&$3713$&$13609$\\
\bottomrule
\end{tabular}
\label{3.TA3}
\end{table}
Both the algorithms described in Section \ref{3.S3} are then applied to determine the genetic ancestry group with the minimum value of $\mu_{TMB_i}$. The results are summarized in Table \ref{3.TA5}.
\begin{table}[h]
\caption{Application results for online and MAB-based algorithm for Cancer-type Stratification}
\centering
\begin{tabular}{cccccc}
\toprule
Algorithm&$\delta^*$&Stages&$N_1$&$N_2$&Extreme Population\\
\midrule
Online&$0.08$&&$530$&$437$&ASJ-EUR\\
MAB-based&$0.08$&$16$&$1038$&$842$&ASJ-EUR\\
\bottomrule
\end{tabular}
\label{3.TA5}
\end{table}
The first column of Table \ref{3.TA5} lists the algorithm and the second column gives the indifference margins. For the MAB-based algorithm, the third column provide the number of stages. The fourth and fifth columns of Table \ref{3.TA5} report the final PSU sample sizes. The final column  indicates the extreme population selected by algorithm.

From Table \ref{3.TA5}, we see that in cancer-type stratification, ASJ-EUR appears as the extreme group. Thus, ASJ-EUR patients consistently have lower TMB scores compared to nonASJ-EUR individuals and thus needs special consideration. This result alings with the study by \cite{taraszka2024comprehensive}, who reported lower TMB in ASJ-EUR individuals relative to nonASJ-EUR individuals in the Profile and Tempus cohorts, providing external support for our result.
To further validate this finding, we compute the population-level value of $\mu_{TMB_i}
$ in \eqref{3.TMB} using the complete dataset, treating it as the full population rather than a sample. The resulting values are $1.667002$ for ASJ-EUR and $1.753239$ for nonASJ-EUR, confirming that ASJ-EUR is indeed the population with the lower mean $\log(TMB+1)$ score. 
Hence, the result obtained from the sample selected through the proposed procedure agrees with that based on the complete population, correctly identifying the group with the lowest mean $\log(TMB+1)$ value while requiring fewer observations. Thus, the proposed procedures correctly identify the same extreme population as the full dataset, while using only a fraction of the observations, that is, $N_1=530, N_2=437$ for the online algorithm and $N_1=1038, N_2=842$ for the MAB-based algorithm, compared to $3713$ and $13609$ total observations for ASJ-EUR and nonASJ-EUR, respectively. This also shows the efficiency of our procedure.

\section{Identifying Anomalous Population}\label{3.S6}
An anomalous population is one whose characteristics of interest differs substantially from that of the remaining populations, typically by being exceptionally high or low. Such populations arise across several domains including those discussed earlier. For example, in economic studies, a region with exceptionally high income inequality relative to other regions may be regarded as an anomalous population. Similarly, in agriculture, a crop variety with extremely low productivity or extremely high productivity compared to the remaining varieties may also be treated as an anomalous population. In clinical trials, an anomalous treatment can be characterized in two ways. It may correspond to a treatment with severe side effects or low therapeutic efficacy, or it may represent a treatment that performs superior relative to the other treatments being compared. Anomaly detection has been studied extensively in the statistics and machine learning literature. For details on anomaly detection, we refer to \cite{dey2026integrating}, \cite{CASTRILLONCANDAS2022104885}, \cite{bouneffouf2020survey}, \cite{chandola2009anomaly} and others. Here, we argue that detecting anomalous population can be formulated directly within our extreme-population selection framework.

Formally, a population is considered anomalous if the value of its characteristic of interest is at least $\delta^*(>0)$ smaller than the minimum, or larger than the maximum, of the corresponding measures of the remaining populations. In the first case, the detection problem is equivalent to the minimum-selection problem described in Subsection \ref{3.S2.1}. In the second case, it reduces to the same framework by considering the negative of the measure of interest. The problem of identifying an anomalous population, whenever one exists, is therefore equivalent to the problem in Subsection \ref{3.S2.1}. As a result, the proposed online algorithm and the MAB-based algorithm can be directly used for identifying the anomalous population.

\section{Concluding Remarks}\label{3.S7}
Identifying the extreme population among a given set of $K(\geq 2)$ populations is an important problem across various fields. The selection is typically carried out using a measure, which may differ depending on the application area. We consider a measure that can be expressed as a continuous real-valued function of an $l$-dimensional parameter vector, possessing first-order partial derivatives. Since true value is unknown, we estimate it using generalized method of moment under multistage sampling design. A multistage sampling design partitions the population into different strata. Each stratum is then divided into primary sampling units (PSUs), which are further partitioned into substrata and each substratum contains the ultimate sampling units. In this article, we develop two sequential procedures, namely an online algorithm and an MAB-based algorithm, for identifying the extreme population. Without imposing parametric assumptions on the underlying distributions, we establish that both algorithms correctly identify the extreme population with a prescribed level of confidence. 

To explore the performance of both the algorithms, we used simulation study to find the population with lower Gini index, a measure of income inequality. The extensive simulation study confirms that the empirical performance of both algorithms is consistent with the theoretical results. In the genetics application, a measure derived from the tumor mutation burden score is considered and the practical applicability of the proposed algorithms is demonstrated using the real-world Memorial Sloan Kettering-IMPACT 50000 clinical sequencing cohort dataset. Both algorithms agree with each other in finding the extreme population obtained from the complete dataset, with both requiring fewer observations than the actual dataset. In fact, the online algorithm used fewer number of observations. This demonstrates that the proposed procedures are both accurate and practically cost-efficient. Finally, we show that the proposed methodologies extends naturally to the problem of detecting an anomalous population, when one exists. Together, these results suggest that the proposed methodologies are broadly applicable for extreme-population selection under realistic, heterogeneous sampling conditions.

\section{Appendix: }
\subsection{Assumptions for the Methodology}
\begin{assumption}\label{3.AN1}
   For each $j=1,2,\ldots,K$, $\widetilde{m}_i^j(\cdot,\bm{\theta}_i)$ is continuous at each $\bm{\theta}_i$ almost surely.
\end{assumption}
\begin{assumption}\label{3.AN2}
     There exists a function $d_i(\cdot)$ satisfying $E(d_i(\cdot))<\infty$ such that, for every $j = 1,2,\ldots,K$ and all $t$, the inequality $\left|\widetilde{m}_i^j(t,\bm{\theta}_i)\right| \le d_i(t)$ holds.
\end{assumption}
\begin{assumption}\label{3.AN3}
 The space $\Theta_i$ of parameters is compact.
\end{assumption}
\begin{assumption}\label{3.AN4}
    $E(\widetilde{\textbf{m}}_i(x,\bm{\theta}_i))$ is continuously differentiable at $\theta_{i0}$ and $\frac{1}{n_i}\sum_{j=1}^{n_i}\frac{\partial}{\partial\bm{\theta}_i}E(\widetilde{\textbf{m}}_{ij}(\bm{\theta}_i))\xrightarrow{a.s.} K_i,$ where $K_i$ is non-singluar matrix.
\end{assumption}
\begin{assumption}\label{3.AN5}
     The sequence $v_{in_i}(\bm{\theta}_i)=\frac{1}{\sqrt{n_i}}\sum_{j=1}^{n_i}\{\widetilde{\textbf{m}}_{ij}(\bm{\theta}_i)-E(\widetilde{\textbf{m}}_{ij}(\bm{\theta}_i))\}$ is stochastically equicontinuous.
\end{assumption}
\begin{assumption}\label{3.AN6} $\text{sup}_{\bm{\theta}_i\in\Theta_i}E|\widetilde{\textbf{m}}_{ij}(\bm{\theta}_i)|^3<\infty.$
\end{assumption}
\begin{assumption}\label{3.AN7}
$\text{lim}_{n_i\rightarrow\infty}\sum_{j=1}^{n_i}Var(\widetilde{\textbf{m}}_{ij}(\bm{\theta}_i))/j^2<\infty \text{ and } \text{lim}_{n_i\rightarrow\infty}\frac{1}{n_i}\sum_{j=1}^{n_i}{Var(\widetilde{\textbf{m}}_{ij}(\bm{\theta}_i))}=W_{i0}$.
\end{assumption}
\begin{assumption}\label{3.AN8}
    $A_{in_i}\xrightarrow{a.s.}A_{i0}$, where $A_{in_i}$ is positive definite matrix with probability 1 for each $n_i$.
\end{assumption}
\begin{assumption}\label{3.AN9}
    Let $s_i,s_i'=1,2,..,S_i, c_{is_i}=1,2,..,n_{is_i}, c'_{is_i'}=1,2,..,n_{is_i'}, b_{ic_{is_i}}, b_{c'_{is_i'}}=1,2 \text{ and } h,h'=1,2,..,k .\\ \text{ If } s_i\neq s_i' \text{ or } c_{is_i}\neq c'_{is_i'} , \{x_{s_ic_{is_i}b_{ic_{is_i}}h}\} \text{ and }  \{x_{s'_ic'_{is_i'}b_{ic'_{is_i'}}h'}\}$ are independent unless they are dependent.
\end{assumption}
\begin{assumption}\label{3.AN10}
    For $s_i \neq s_i'$, $\{x_{s_i c_{i s_i} b_{i c_{i s_i}} h}\}$ and $\{x_{s'_i c'_{i s'_i} b_{i c_{i s'_i}} h'}\}$ are not required to follow the same distribution.
\end{assumption}
\begin{assumption}\label{3.AN11}
    The PSU-level variables are identically distributed across different PSUs within each stratum.
\end{assumption}
\begin{assumption}\label{3.AN16}
 $\mathbb{E}[|X_i| \mid s_i] < \infty$.
\end{assumption}
\begin{assumption}\label{3.AN12}
  Within each stratum, the stratum-level variables have finite variance.
\end{assumption}
\begin{assumption}\label{3.AN13}
    The $(2+\delta)^{th}$ moments of the PSU totals exist for some $\delta>0$ and are uniformly bounded above by $B_\delta$.
\end{assumption}
\begin{assumption}\label{3.AN14}
    For each $p \in [0,1]$, both $F_i(x)$ and its derivative $\frac{dF_i(x)}{dx}$ are continuous for all $x$ in a neighborhood $J$ of the $p^{th}$ population quantile $z(p)$.
\end{assumption}
\begin{assumption}\label{3.AN15}
For some constant $0 < C < \infty$,
$
Var\big\{\widehat{F}_{in_i}(x+\delta) - \widehat{F}_{in_i}(x)\big\}
\leq C n_i^{-1} |\delta|,
$
for all sample sizes $n_i$, whenever both $x$ and $x+\delta$ belong to the neighborhood $J$.
\end{assumption}
\noindent The assumptions \ref{3.AN1}-\ref{3.AN16} are taken from \cite{BD2005}.
\renewcommand{\theequation}{A.\arabic{equation}} 
\subsection{Proof of Lemmas}\label{3.8.2}
\begin{lemma}\label{3.L1}
If $g\left(\bm{\theta}\right):\mathbb{R}^p\rightarrow \mathbb{R}$ is a function having continuous first-order partial derivative, then, as $n_i\rightarrow\infty$,
$\sqrt{n_i}\left(g\left(\bm{\widehat{\theta}}_{in_i}\right)-g\left(\bm{\theta}_i\right)\right)\xrightarrow{d}N\left(0,\xi_i^2\right),$
where $\xi_i^2$ is asymptotic variance of $\sqrt{n_i}\left(g\left(\bm{\widehat{\theta}}_{in_i}\right)\right)$ and $\bm{\widehat{\theta}}_{in_i}$ is GMM estimator of $\bm{\theta}_i$.
\end{lemma}
\begin{proof}{Proof}
Since $g$ is a continuous function from $\mathbb{R}^p $ to $\mathbb{R}$, the Taylor's series expansion around $\bm{\theta}_i=(\theta_{i1},\theta_{i2},..,\theta_{il})$ is given by  
$g\left(\bm{\theta}\right)=g\left(\bm{\theta}_i\right)+\sum_{j=1}^{l}\frac{\partial g\left(\bm{\theta}\right)}{\partial \theta_{ij}}\bigg|_{\bm{\theta}_i+p'(\bm{\theta}-\bm{\theta}_i)}(\theta_{j}-\theta_{ij}).
$
Consequently, we get
\begin{eqnarray}\label{3.7}
\sqrt{n_i}\left(g\left(\bm{\widehat{\theta}}_{in_i}\right)-g\left(\bm{\theta}_i\right)\right)=\sum_{j=1}^{l}\frac{\partial g\left(\bm{\theta}\right)}{\partial \theta_{ij}}\bigg|_{\bm{\theta}_i+p'(\bm{\widehat{\theta}}_{in_i}-\bm{\theta}_i)}\sqrt{n_i}(\widehat{\theta}_{ijn_i}-\theta_{ij}).
\end{eqnarray}
Consider the $1\times l$ matrix $\bar{H}_i$ whose $j^{th}$ entry is $\frac{\partial g\left(\bm{\theta}\right)}{\partial \theta_{ij}}\bigg|_{\bm{\theta}_i+p'(\bm{\widehat{\theta}}_{in_i}-\bm{\theta}_i)}$. Since the first-order partial derivatives of $g$ are continuous and $\widehat{\bm{\theta}}_{in_i}\xrightarrow{P}\bm{\theta}_i$, it follows from the continuous mapping theorem that
$
\bar{H}_i\xrightarrow{P}H_i,
$
where $H_i$ is $1\times p$ matrix whose $j^{th}$ entry is $\frac{\partial g\left(\bm{\theta}\right)}{\partial \theta_{ij}}\bigg|_{\bm{\theta}_i}.$
Consequently, using equation \eqref{3.2} and Slutsky theorem, we get
$\bar{H}_i\left(\sqrt{n_i}\left(\widehat{\bm\theta}_{in_i}-\bm{\theta}_i\right)\right)\xrightarrow{d}N\left(0,\xi_i^2\right),$
where $\xi_i^2=H_i'\Sigma_iH_i$.
Therefore, using equation \eqref{3.7}, it follows that
$
\sqrt{n_i}\left(g\left(\bm{\widehat{\theta}}_{in_i}\right)-g\left(\bm{\theta}_i\right)\right)\xrightarrow{d}N\left(0,\xi_i^2\right).
$
Hence the result.
\end{proof}
\begin{lemma}\label{3.L2}
If $Z_n=f\left(X_n,Y_n\right)$ is a sequence of random variable such that $Z_n|Y_n\xrightarrow{d}N(0,1)$ and $Y_n\xrightarrow{P}c,$ where $c>0$ is a constant, then
$Z_n\xrightarrow{d}N(0,1).$
\end{lemma}
\begin{proof}{Proof}
For $x\in\mathbb{R}$, define the event $A_n=\{Z_n\leq x\}$. Then,
$P(A_n)=P(Z_n\leq x)=E(P(Z_n\leq x|{Y_n})).$
Now, for any $\epsilon>0$, we write the above expectation as
\begin{eqnarray*}
E(P(Z_n\leq x|{Y_n}))&=&E(P(Z_n\leq x|{Y_n})\textbf{1}(|Y_n-c|\leq\epsilon))+E(P(Z_n\leq x|{Y_n})\textbf{1}(|Y_n-c|>\epsilon))\\&\leq& E(\Phi(x)\textbf{1}(|Y_n-c|\leq\epsilon))+E((P(Z_n\leq x|{Y_n})-\Phi(x))\textbf{1}(|Y_n-c|\leq\epsilon))\\&&+P(|Y_n-c|>\epsilon)
\\&=& \Phi(x)P(|Y_n-c|\leq\epsilon)+E((P(Z_n\leq x|{Y_n})-\Phi(x))\textbf{1}(|Y_n-c|\leq\epsilon))\\&&+ P(|Y_n-c|>\epsilon).
\end{eqnarray*}
Since
$
E\left((P(Z_n\leq x\mid Y_n)-\Phi(x))\textbf{1}(|Y_n-c|\leq\epsilon)\right)\leq\sup_{|Y_n-c|\leq\epsilon}|P(Z_n\leq x|Y_n)-\Phi(x)|,$
$Y_n\xrightarrow{P}c$ and $Z_n|Y_n\xrightarrow{d}N(0,1)$, from the above inequality, we get
$
P(Z_n\leq x)\rightarrow\Phi(x).
$
Hence,
$
Z_n\xrightarrow{d}N(0,1).
$
\end{proof}
\begin{lemma}\label{3.L3}
For $i=1,2,\ldots,K$, show that if $n_i>n_0$, $n_i \geq \frac{V_{in_i}^2}{t}$ and $n_0,n_i,n_i-n_0\rightarrow\infty$ as $t\rightarrow0$, then
\[\frac{\widetilde{g}\left(\widehat{\bm{\theta}}_{in_0},\widehat{\bm{\theta}}_{i(n_i-n_0)}\right)-g\left(\bm{\theta}_i\right)}{\sqrt{t}}\xrightarrow{d}N(0,1),\;\widetilde{g}\left(\widehat{\bm{\theta}}_{in_0},\widehat{\bm{\theta}}_{i(n_{i}-n_0)}\right)=c_{i} g\left(\widehat{\bm{\theta}}_{in_0}\right)+(1-c_{i})g\left(\widehat{\bm{\theta}}_{i\left(n_{i}-n_0\right)}\right),\]
where
$c_{i}=\left(n_0/n_{i}\right)\left(1+\sqrt{1-\left(n_{i}/n_0\right)\left(1-\left(\left(n_{i}-n_0\right)t/V_{in_{i}}^2\right)\right)}\right),$
and $V_{in_i}^2\xrightarrow{P}\xi_i^2$.
\end{lemma}
\begin{proof}{Proof}
Using Lemma \ref{3.L1}, we get
\begin{eqnarray}\label{3.8}
\sqrt{n_0}\left(g\left(\bm{\widehat{\theta}}_{in_{0}}\right)-g\left(\bm{\theta}_i)\right)\right)\xrightarrow{d}N\left(0,\xi_i^2\right),\;\text{and } \sqrt{n_i-n_0}\left(g\left(\bm{\widehat{\theta}}_{i(n_i-n_0)}\right)-g\left(\bm{\theta}_i\right)\right)\xrightarrow{d}N\left(0,\xi_i^2\right).
\end{eqnarray}
If $V_{in_i}^2$ is known, then $c_i$ becomes a constant. Since PSUs are independent to each other, using equation \eqref{3.8}, we get
$\frac{\left(c_ig\left(\bm{\widehat{\theta}}_{in_{0}}\right)+(1-c_i)g\left(\bm{\widehat{\theta}}_{i(n_i-n_0)}\right)-g\left(\bm{\theta}_i\right)\right)}{\sqrt{\frac{c_i^2\xi_i^2}{n_0}+\frac{(1-c_i)^2\xi_i^2}{n_i-n_0}}}\Bigg|_{V_{in_i}^2}\xrightarrow{d}N\left(0,1\right).$
Since $V_{in_i}^2\xrightarrow{P}\xi_i^2$, using Lemma \ref{3.L2}, we get
$\frac{c_ig\left(\bm{\widehat{\theta}}_{in_{0}}\right)+(1-c_i)g\left(\bm{\widehat{\theta}}_{i(n_i-n_0)}\right)-g\left(\bm{\theta}_i\right)}{\sqrt{\frac{c_i^2\xi_i^2}{n_0}+\frac{(1-c_i)^2\xi_i^2}{n_i-n_0}}}\xrightarrow{d}N\left(0,1\right).$
Consequently, using $V_{in_i}^2\xrightarrow{P}\xi_i^2$ and Slutsky theorem, we get
$\frac{c_ig\left(\bm{\widehat{\theta}}_{in_{0}}\right)+(1-c_i)g\left(\bm{\widehat{\theta}}_{i(n_i-n_0)}\right)-g\left(\bm{\theta}_i\right)}{V_{in_i}\sqrt{\frac{c_i^2}{n_0}+\frac{(1-c_i)^2}{n_i-n_0}}}\xrightarrow{d}N\left(0,1\right).$
From this, we conclude that
$\frac{\widetilde{g}\left(\widehat{\bm{\theta}}_{in_0},\widehat{\bm{\theta}}_{i(n_i-n_0)}\right)-g\left(\bm{\theta}_i\right)}{\sqrt{t}}\xrightarrow{d}N(0,1).$
\end{proof}
\begin{lemma}\label{3.L4}
For $i^{th}$ population and given $l>0$, let $N_i(\leq H_i)$ denote the smallest integer   
$n_i(\geq n_{0})$, for which
\begin{eqnarray}\label{3.10}
n_{i}\geq\frac{1}{t}\left(V_{in_{i}}^2+\frac{1}{n_{i}^l}\right)=\widehat{C}_{i} \text{ and } n_{is_i}\geq\widehat{C}_{is_i}=\widehat{C}_{i}a_{is_i}, \text{ for all } s_i=1,2,..,S_i,
\end{eqnarray}
where $t>0$ is some constant.
For each $i$, based on final PSU size $N_{i}$, define $c_{i1}$ and $c_{i2}$ as follows
\begin{equation}\label{3.11}
  c_{ij} =
  \begin{cases}
    1, & \text{if } M_{ij}=n_0 \\\\
    \frac{n_0}{M_{ij}}\left(1+\sqrt{1-\frac{M_{ij}}{n_0}\left(1-\frac{(M_{ij}-n_0)t}{V_{iM_{ij}}^2}\right)}\right), &\text{ if }M_{ij}\neq n_0,
  \end{cases}
\end{equation}
\begin{eqnarray}\label{3.13}
n_0=min\left\{H_i,max\left\{2,\left\lceil\left(\frac{1}{t}\right)^{1/(l+1)}\right\rceil\right\}\right\},
\;
n_{0s_i}=min\left\{H_{is_i},max\left\{2,\left\lceil\left(\frac{1}{t}\right)^{1/(l+1)}\times a_{is_i}\right\rceil\right\}\right\},
\end{eqnarray}
where $M_{i1}=N_i$, $M_{i2}=C_i$, $C_i=\left\lceil\frac{\xi_i^2}{t}\right\rceil$ and $V_{in_i}^2$ is consistent estimator of $\xi_i^2$. Then, as $t\rightarrow0$, following holds
 $(i)$ $\frac{c_{ij}}{\sqrt{n_0t}}\xrightarrow{P}0,$ for $j=1,2,\;$
    $(ii)$ $\frac{c_{ij}}{\sqrt{(C_i-n_0)t}}\xrightarrow{P}0,$ for $j=1,2,\;$
    $(iii)$ $\frac{c_{i1}}{\sqrt{(N_i-n_0)t}}\xrightarrow{P}0.$
\end{lemma}
\begin{proof}{Proof}
    For sufficiently small $t$, $N_i\neq n_0$ and $C_i\neq n_0.$ Using equations \eqref{3.10}-\eqref{3.13}, for each $i$, we get
    \begin{eqnarray}\label{3.14}
        N_it\xrightarrow{P}\xi_i^2, \frac{n_0}{N_i}\xrightarrow{P}0, \frac{N_it}{V_{iN_i}^2}\xrightarrow{P}1, C_it=\xi_i^2, \frac{n_0}{C_i}\rightarrow0, \text{ and  }\frac{C_it}{V_{iC_i}^2}\xrightarrow{P}1.
    \end{eqnarray}
    \textbf{Proof. \textit{(i)}} For $j=1,2,$ using equation \eqref{3.11}, we get
    $\frac{c_{ij}}{\sqrt{n_0t}}=\frac{1}{\sqrt{M_{ij}t}}\left(\sqrt{\frac{n_0}{M_{ij}}}+\sqrt{\frac{n_0}{M_{ij}}-\left(1-\frac{(M_{ij}-n_0)t}{V_{iM_{ij}}^2}\right)}\right).$
    Consequently, by using equation \eqref{3.14} and Slutsky theorem, we get
    $c_{ij}/\sqrt{n_0t}\xrightarrow{P}0.$\\
    \textbf{Proof. \textit{(ii)}} For $j=1,2,$ since 
    $c_{ij}/\sqrt{(C_i-n_0)t}=c_{ij}/\left(\sqrt{\frac{C_i}{n_0}-1}\sqrt{n_0t}\right),$
    by using part \textit{(i)}, equation \eqref{3.14} and Slutsky thoerem, we get
    $c_{i1}/\sqrt{(C_i-n_0)t}\xrightarrow{P}0.$\\
    \textbf{Proof. \textit{(iii)}} Since 
    $c_{i1}/\sqrt{(N_i-n_0)t}=c_{i1}/\left(\sqrt{\frac{N_i}{n_0}-1}\sqrt{n_0t}\right),$
    by using part \textit{(i)}, equation \eqref{3.14} and Slutsky thoerem, we get
    $c_{i1}/\sqrt{(N_i-n_0)t}\xrightarrow{P}0.$
\end{proof}
\begin{lemma}\label{3.L5}
Under the conditions \eqref{3.10}-\eqref{3.13} provided in Lemma \ref{3.L4}, if $g\left(\bm{\widehat{\theta}}_{in_i}\right)$ satisfies uniform continuity in probability condition, then 
$\frac{\widetilde{g}\left(\widehat{\bm{\theta}}_{in_0},\widehat{\bm{\theta}}_{i(N_i-n_0)}\right)-g\left(\bm{\theta}_i\right)}{\sqrt{t}}\xrightarrow{d}N(0,1).$
\end{lemma}
\begin{proof}{Proof} For every $t>0$, we have
\begin{eqnarray*}
\frac{\widetilde{g}\left(\widehat{\bm{\theta}}_{in_0},\widehat{\bm{\theta}}_{i(N_i-n_0)}\right)-g\left(\bm{\theta}_0\right)}{\sqrt{t}}=\frac{\widetilde{g}\left(\widehat{\bm{\theta}}_{in_0},\widehat{\bm{\theta}}_{i(N_i-n_0)}\right)-\widetilde{g}\left(\widehat{\bm{\theta}}_{in_0},\widehat{\bm{\theta}}_{i(C_i-n_0)}\right)}{\sqrt{t}}+\frac{\widetilde{g}\left(\widehat{\bm{\theta}}_{in_0},\widehat{\bm{\theta}}_{i(C_i-n_0)}\right)-g\left(\bm{\theta}_i\right)}{\sqrt{t}}.
\end{eqnarray*}
From Lemma \ref{3.L3}, we have
\begin{eqnarray}\label{3.15}
\frac{\widetilde{g}\left(\widehat{\bm{\theta}}_{in_0},\widehat{\bm{\theta}}_{i(C_i-n_0)}\right)-g\left(\bm{\theta}_i\right)}{\sqrt{t}}\xrightarrow{d} N(0,1).
\end{eqnarray}
Now, we prove that 
$
\frac{\widetilde{g}\left(\widehat{\bm{\theta}}_{in_0},\widehat{\bm{\theta}}_{i(N_i-n_0)}\right)-\widetilde{g}\left(\widehat{\bm{\theta}}_{in_0},\widehat{\bm{\theta}}_{i(C_i-n_0)}\right)}{\sqrt{t}}\xrightarrow{P}0.
$
Let $\widetilde{g}\left(\widehat{\bm{\theta}}_{in_0},\widehat{\bm{\theta}}_{i(N_i-n_0)}\right)=\widetilde{g}_{N_i}$ and $\widetilde{g}\left(\widehat{\bm{\theta}}_{in_0},\widehat{\bm{\theta}}_{i(C_i-n_0)}\right)=\widetilde{g}_{C_i}$. Using the expression for $\widetilde{g}$, we get
\[
\widetilde{g}_{N_i}-\widetilde{g}_{C_i}
=(c_{i1}-c_{i2})g\left(\bm{\widehat{\theta}}_{in_{0}}\right)+(1-c_{i1})g\left(\bm{\widehat{\theta}}_{i(N_i-n_0)}\right)-(1-c_{i2})g\left(\bm{\widehat{\theta}}_{i(C_i-n_0)}\right).
\]
On adding and subtracting $(1-c_{i1})g\left(\bm{\widehat{\theta}}_{i(C_i-n_0)}\right)$, we get
\begin{eqnarray*}
\widetilde{g}_{N_i}-\widetilde{g}_{C_i}&=&
(c_{i1}-c_{i2})\left(g\left(\bm{\widehat{\theta}}_{in_{0}}\right)-g\left(\bm{\theta}_i\right)\right)+(1-c_{i1})\left(g\left(\bm{\widehat{\theta}}_{i(N_i-n_0)}\right)-g\left(\bm{\widehat{\theta}}_{i(C_i-n_0)}\right)\right)\\&&+(c_{i2}-c_{i1})\left(g\left(\bm{\widehat{\theta}}_{i(C_i-n_0)}\right)-g\left(\bm{\theta}_i\right)\right).
\end{eqnarray*}
Now, divide both sides by $\sqrt{t}$, we get
\begin{eqnarray*}
\frac{\widetilde{g}_{N_i}-\widetilde{g}_{C_i}}{\sqrt{t}}&=&\frac{(c_{i1}-c_{i2})\left(g\left(\bm{\widehat{\theta}}_{in_{0}}\right)-g\left(\bm{\theta}_i\right)\right)}{\sqrt{t}}+\frac{(1-c_{i1})\left(g\left(\bm{\widehat{\theta}}_{i(N_i-n_0)}\right)-g\left(\bm{\widehat{\theta}}_{i(C_i-n_0)}\right)\right)}{\sqrt{t}}\\&&+\frac{(c_{i2}-c_{i1})\left(g\left(\bm{\widehat{\theta}}_{i(C_i-n_0)}\right)-g\left(\bm{\theta}_i\right)\right)}{\sqrt{t}}.
\end{eqnarray*}
Consequently, we get
\begin{eqnarray*}
\frac{\widetilde{g}_{N_i}-\widetilde{g}_{C_i}}{\sqrt{t}}&=&\frac{(c_{i1}-c_{i2})\sqrt{n_0}\left(g\left(\bm{\widehat{\theta}}_{in_{0}}\right)-g\left(\bm{\theta}_i\right)\right)}{\sqrt{n_0t}}+\frac{(1-c_{i1})\sqrt{N_i-n_0}\left(g\left(\bm{\widehat{\theta}}_{i(N_i-n_0)}\right)-g\left(\bm{\widehat{\theta}}_{i(C_i-n_0)}\right)\right)}{\sqrt{(N_i-n_0)t}}\\&&+\frac{(c_{i2}-c_{i1})\sqrt{C_i-n_0}\left(g\left(\bm{\widehat{\theta}}_{i(C_i-n_0)}\right)-g\left(\bm{\theta}_i\right)\right)}{\sqrt{(C_i-n_0)t}}.
\end{eqnarray*}
Consequently, using Lemma \ref{3.L1}, Lemma \ref{3.L4}, uniform continuity in probability condition and Slutsky theorem, we get
\begin{eqnarray}\label{3.16}
\frac{\widetilde{g}\left(\widehat{\bm{\theta}}_{in_0},\widehat{\bm{\theta}}_{i(N_i-n_0)}\right)-\widetilde{g}\left(\widehat{\bm{\theta}}_{in_0},\widehat{\bm{\theta}}_{i(C_i-n_0)}\right)}{\sqrt{t}}=\frac{\widetilde{g}_{N_i}-\widetilde{g}_{C_i}}{\sqrt{t}}\xrightarrow{P}0.
\end{eqnarray}
Consequently, using equations \eqref{3.15}, \eqref{3.16} and Slutsky theorem, we conclude that
\[\frac{\widetilde{g}\left(\widehat{\bm{\theta}}_{in_0},\widehat{\bm{\theta}}_{i(N_i-n_0)}\right)-g\left(\bm{\theta}_i\right)}{\sqrt{t}}=\frac{\widetilde{g}_{N_i}-\widetilde{g}_{C_i}}{\sqrt{t}}+\frac{\widetilde{g}\left(\widehat{\bm{\theta}}_{in_0},\widehat{\bm{\theta}}_{i(C_i-n_0)}\right)-g\left(\bm{\theta}_i\right)}{\sqrt{t}}\xrightarrow{d} N(0,1).\]
\end{proof}
\begin{lemma}\label{3.L6}
Show that
\[\lim_{r\rightarrow\infty}P\left\{ -z_{(\alpha/2K)}<\frac{\widetilde{g}\left(\widehat{\bm{\theta}}_{in_0},\widehat{\bm{\theta}}_{i(N_{ir}-n_0)}\right)-g\left(\bm{\theta}_i\right)}{\sqrt{t_r}}<z_{(\alpha/2K)}, i=1,2,...,K\right\}\geq 1-\alpha.\]
\end{lemma}
\begin{proof}{Proof}
Since populations are independent to each other, we have
\begin{eqnarray*}
&&P\left\{ \left|\frac{\widetilde{g}\left(\widehat{\bm{\theta}}_{in_0},\widehat{\bm{\theta}}_{i(N_{ir}-n_0)}\right)-g\left(\bm{\theta}_i\right)}{\sqrt{t_r}}\right|<z_{(\alpha/2K)}, i=1,2,...,K\right\}\\&&\hspace{3cm}=\prod_{i=1}^{K}P\left\{ \left|\frac{\widetilde{g}\left(\widehat{\bm{\theta}}_{in_0},\widehat{\bm{\theta}}_{i(N_{ir}-n_0)}\right)-g\left(\bm{\theta}_i\right)}{\sqrt{t_r}}\right|<z_{(\alpha/2K)}\right\}.
\end{eqnarray*}
Consequently, by using Lemma \ref{3.L5}, we get
\begin{eqnarray*}
    \lim_{r\rightarrow\infty}P\left\{ -z_{(\alpha/2K)}<\frac{\widetilde{g}\left(\widehat{\bm{\theta}}_{in_0},\widehat{\bm{\theta}}_{i(N_{ir}-n_0)}\right)-g\left(\bm{\theta}_i\right)}{\sqrt{t_r}}<z_{(\alpha/2K)}, i=1,2,...,K\right\}=\left(1-\frac{\alpha}{K}\right)^K.
\end{eqnarray*}
Since $\left(1-\alpha/K\right)^K>1-\alpha$, from above equation, we get
\[\lim_{r\rightarrow\infty}P\left\{ -z_{(\alpha/2K)}<\frac{\widetilde{g}\left(\widehat{\bm{\theta}}_{in_0},\widehat{\bm{\theta}}_{i(N_{ir}-n_0)}\right)-g\left(\bm{\theta}_i\right)}{\sqrt{t_r}}<z_{(\alpha/2K)}, i=1,2,...,K\right\}\geq1-\alpha.\]
\end{proof}
\begin{lemma}\label{3.L8}
Prove that for all $r\in\mathbb{N}$ and $\epsilon>0$, there exists $\delta_0(r,\epsilon)>0$ such that for all $\delta^*<\delta_0$, we have\\\\
\resizebox{\textwidth}{!}{%
$\left|\sum_{j=2}^{K}P\left\{\frac{\widetilde{g}\left(\widehat{\bm{\theta}}_{(j)n_0},\widehat{\bm{\theta}}_{(j)(N_{(j)r}-n_0)}\right)-g\left(\bm{\theta}_j\right)}{\sqrt{t_r}}<\frac{\widetilde{g}\left(\widehat{\bm{\theta}}_{(1)n_0},\widehat{\bm{\theta}}_{(1)(N_{(1)r}-n_0)}\right)-g\left(\bm{\theta}_1\right)}{\sqrt{t_r}}-\frac{\delta^*}{\sqrt{t_r}}-2z_{\frac{\alpha}{2K}}\right\}-\frac{\delta}{2^{r+b}}\right|<\epsilon,$%
}
where $t_r=\left(\frac{\delta^*}{h_r-2z_{(\alpha/2K)}}\right)^2$ and $h_r$ satisfy the following equation
$
(K-1)\int_{-\infty}^{\infty}\Phi\left(y-h_r\right)\phi(y)dy=\frac{\delta}{2^{r+b}}.
$
Here $\delta>0$ and $b$ is chosen in such a way that $h_r>2z_{(\alpha/2K)}$ for all $r\in \mathbb{N}$.
\end{lemma}
\begin{proof}{Proof}
For $r\in\mathbb{N}$, using Lemma \ref{3.L5}, we have\\
\resizebox{\textwidth}{!}{%
$\lim_{\delta^*\rightarrow0}\sum_{j=2}^{K}P\bigg\{\frac{\widetilde{g}\left(\widehat{\bm{\theta}}_{(j)n_0},\widehat{\bm{\theta}}_{(j)(N_{(j)r}-n_0)}\right)-g\left(\bm{\theta}_j\right)}{\sqrt{t_r}}<\frac{\widetilde{g}\left(\widehat{\bm{\theta}}_{(1)n_0},\widehat{\bm{\theta}}_{(1)(N_{(1)r}-n_0)}\right)-g\left(\bm{\theta}_1\right)}{\sqrt{t_r}}-\frac{\delta^*}{\sqrt{t_r}}-2z_{\frac{\alpha}{2K}}\bigg\}$%
}
\[=\sum_{j=2}^{K}\int_{-\infty}^{\infty}\left(\Phi\left(y-h_r\right)\right)\phi(y)dy=\frac{\delta}{2^{r+b}}.\]
This implies for all $r\in\mathbb{N}$ and all $\epsilon>0$, there exists $\delta_0(r,\epsilon)>0$ such that for all $\delta^*<\delta_0$, we have\\
\resizebox{\textwidth}{!}{%
$\left|\sum_{j=2}^{K}P\left\{\frac{\widetilde{g}\left(\widehat{\bm{\theta}}_{(j)n_0},\widehat{\bm{\theta}}_{(j)(N_{(j)r}-n_0)}\right)-g\left(\bm{\theta}_j\right)}{\sqrt{t_r}}<\frac{\widetilde{g}\left(\widehat{\bm{\theta}}_{(1)n_0},\widehat{\bm{\theta}}_{(1)(N_{(1)r}-n_0)}\right)-g\left(\bm{\theta}_1\right)}{\sqrt{t_r}}-\frac{\delta^*}{\sqrt{t_r}}-2z_{\frac{\alpha}{2K}}\right\}-\frac{\delta}{2^{r+b}}\right|<\epsilon.$%
}
\\\\Hence the result.
\end{proof}
\begin{lemma}\label{3.L9}
If for all $j\in\{2,3,..,K\}$,
\[
\sum_{n=n_0+2}^{\infty}\sum_{m=n_0+2}^{\infty}P\left(\widetilde{g}\left(\widehat{\bm{\theta}}_{(j)n_0},\widehat{\bm{\theta}}_{(j)(m-n_0)}\right)<\widetilde{g}\left(\widehat{\bm{\theta}}_{(1)n_0},\widehat{\bm{\theta}}_{(1)(n-n_0)}\right)\right)<\infty,
\]
then prove that
$\lim_{\delta^*\rightarrow0}\sum_{r=1}^{\infty}\sum_{j=2}^{K}A_{r,j}=\sum_{r=1}^{\infty}\left(\delta/2^{r+b}\right),$
where 
\begin{eqnarray*}
A_{r,j}=P\Bigg\{\frac{\widetilde{g}\left(\widehat{\bm{\theta}}_{(j)n_0},\widehat{\bm{\theta}}_{(j)(N_{(j)r}-n_0)}\right)-g\left(\bm{\theta}_j\right)}{\sqrt{t_r}}<\frac{\widetilde{g}\left(\widehat{\bm{\theta}}_{(1)n_0},\widehat{\bm{\theta}}_{(1)(N_{(1)r}-n_0)}\right)-g\left(\bm{\theta}_1\right)}{\sqrt{t_r}}-\frac{\delta^*}{\sqrt{t_r}}-2z_{\frac{\alpha}{2K}}\Bigg\}.
\end{eqnarray*}
\end{lemma}
\begin{proof}{Proof}{Proof}
Let $\epsilon>0$ be given. Now, using Lemma \ref{3.L8} and assumption, choose $R\in \mathbb{N}$ such that 
\begin{eqnarray}\label{3.18}
\left|\sum_{r=1}^{R}\sum_{j=2}^{K}A_{r,j}-\sum_{r=1}^{R}\frac{\delta}{2^{r+b}}\right|<\frac{\epsilon}{3},\;\left|\sum_{r=R+1}^{\infty}\sum_{j=2}^{K}A_{r,j}\right|<\frac{\epsilon}{3},\text{ and } \left|\sum_{r=R+1}^{\infty}\frac{\delta}{2^{r+b}}\right|<\frac{\epsilon}{3}. 
\end{eqnarray}
Now, consider
\begin{eqnarray}
\left|\sum_{r=1}^{\infty}\sum_{j=2}^{K}A_{r,j}-\sum_{r=1}^{\infty}\frac{\delta}{2^{r+b}}\right|
\leq\left|\sum_{r=1}^{R}\sum_{j=2}^{K}A_{r,j}-\sum_{r=1}^{R}\frac{\delta}{2^{r+b}}\right|+\left|\sum_{r=R+1}^{\infty}\sum_{j=2}^{K}A_{r,j}\right|+\left|\sum_{r=R+1}^{\infty}\frac{\delta}{2^{r+b}}\right|.
\end{eqnarray}
Consequently, using equation \eqref{3.18}, there exists $\delta_1>0$, such that for all $\delta^*<\delta_1$, we have\\
$\left|\sum_{r=1}^{\infty}\sum_{j=2}^{K}A_{r,j}-\sum_{r=1}^{\infty}\frac{\delta}{2^{r+b}}\right|<\epsilon.$
Hence,
$\lim_{\delta^*\rightarrow0}\sum_{r=1}^{\infty}\sum_{j=2}^{K}A_{r,j}=\sum_{r=1}^{\infty}\delta/2^{r+b}.$
\end{proof}
\bibliographystyle{apa}
\bibliography{name-6}

@article{dey2026integrating,
author = {Asim K. Dey and Yulia R. Gel},
title = {Integrating Statistical Data Depth and Topological Methods for Anomaly Detection in Weighted Dynamic Networks},
journal = {Technometrics},
volume = {0},
number = {0},
pages = {1--14},
year = {2026},
publisher = {Taylor \& Francis},
doi = {10.1080/00401706.2026.2643214}
}

@article{chattopadhyay2022minimum,
  title={Minimum cost-compression risk in principal component analysis},
  author={Chattopadhyay, Bhargab and Banerjee, Swarnali},
  journal={Australian \& New Zealand Journal of Statistics},
  volume={64},
  number={4},
  pages={422--441},
  year={2022},
  publisher={Wiley Online Library},
  doi={10.1111/anzs.12378}
}

@article{chandola2009anomaly,
  title={Anomaly detection: A survey},
  author={Chandola, Varun and Banerjee, Arindam and Kumar, Vipin},
  journal={ACM computing surveys (CSUR)},
  volume={41},
  number={3},
  pages={1--58},
  year={2009},
  publisher={ACM New York, NY, USA},
  doi={10.1145/1541880.1541882}
}

@article{horrace2008ranking,
  title={Ranking inequality: Applications of multivariate subset selection},
  author={Horrace, William C and Marchand, Joseph T and Smeeding, Timothy M},
  journal={The Journal of Economic Inequality},
  volume={6},
  number={1},
  pages={5--32},
  year={2008},
  publisher={Springer},
  doi={10.1007/s10888-006-9043-7}
}

@article{taraszka2024comprehensive,
  title={A comprehensive analysis of clinical and polygenic germline influences on somatic mutational burden},
  author={Taraszka, Kodi and Groha, Stefan and King, David and Tell, Robert and White, Kevin and Ziv, Elad and Zaitlen, Noah and Gusev, Alexander},
  journal={The American Journal of Human Genetics},
  volume={111},
  number={2},
  pages={242--258},
  year={2024},
  publisher={Elsevier},
  doi={10.1016/j.ajhg.2023.12.010}
}

@INPROCEEDINGS{bouneffouf2020survey,
  author={Bouneffouf, Djallel and Rish, Irina and Aggarwal, Charu},
  booktitle={2020 IEEE Congress on Evolutionary Computation (CEC)}, 
  title={Survey on Applications of Multi-Armed and Contextual Bandits}, 
  year={2020},
  volume={},
  number={},
  pages={1-8},
  doi={10.1109/CEC48606.2020.9185782}}

@article{samstein2019tumor,
  title={Tumor mutational load predicts survival after immunotherapy across multiple cancer types},
  author={Samstein, Robert M. and Lee, Chang-Han and Shoushtari, Alexander N. and Hellmann, Matthew D. and Shen, Riqiang and Janjigian, Yelena Y. and Barron, David A. and Zehir, Ahmet and Jordan, Eytan J. and Omuro, Antonio and others},
  journal={Nature Genetics},
  volume={51},
  number={2},
  pages={202--206},
  year={2019},
  publisher={Nature Publishing Group US New York},
  doi={10.1038/s41588-018-0312-8},
}

@article{rajput2023evaluation,
  title={Evaluation of a decided sample size in machine learning applications},
  author={Rajput, Daniyal and Wang, Wei-Jen and Chen, Chun-Chuan},
  journal={BMC bioinformatics},
  volume={24},
  number={1},
  pages={48},
  year={2023},
  publisher={Springer},
  doi={10.1186/s12859-023-05156-9}
}

@article{faber2014sample,
  title={How sample size influences research outcomes},
  author={Faber, Jorge and Fonseca, Lilian Martins},
  journal={Dental press journal of orthodontics},
  volume={19},
  pages={27--29},
  year={2014},
  publisher={SciELO Brasil},
  doi={10.1590/2176-9451.19.4.027-029.ebo}
}

@article{https://doi.org/10.1002/pst.70023,
author = {Siriwardhana, Chathura and Gunaratnam, Bakeerathan and Kulasekera, K. B.},
title = {Personalized Treatment Selection for Multivariate Ordinal Scale Outcomes and Multiple Treatments},
journal = {Pharmaceutical Statistics},
volume = {24},
number = {4},
pages = {e70023},
doi = {10.1002/pst.70023},
year = {2025}
}

@article{gebeyaw2024participatory,
author = {Mekonnen Gebeyaw and Asnake Fikre and Alemu Abate and Tesfahun Alemu},
  title = {Participatory variety evaluation and selection of chickpea (Cicer arietinum L.) varieties; an underpinning to novel technology uptake in northwestern Ethiopia},
journal = {Heliyon},
volume = {10},
number = {8},
pages = {e29801},
year = {2024},
doi = {10.1016/j.heliyon.2024.e29801}
}

@Article{agronomy12040765,
AUTHOR = {Rayburn, Edward B. and Basden, Tom},
TITLE = {Comparison of Crop Yield Estimates Obtained from an Historic Expert System to the Physical Characteristics of the Soil Components—A Project Report},
JOURNAL = {Agronomy},
VOLUME = {12},
YEAR = {2022},
NUMBER = {4},
ARTICLE-NUMBER = {765},
ISSN = {2073-4395},
DOI = {10.3390/agronomy12040765}
}

@article{bechhofer1954single,
  title={A single-sample multiple decision procedure for ranking means of normal populations with known variances},
  author={Bechhofer, Robert E},
  journal={The Annals of Mathematical Statistics},
  pages={16--39},
  volume={25(1)},
  year={1954},
  publisher={JSTOR},
  doi={10.1214/aoms/1177728845}
}

@article{tanwar2023understanding,
  title={Understanding the impact of population and cancer type on tumor mutation burden scores: A comprehensive whole-exome study in cancer patients from India},
  author={Tanwar, Nishtha AjitSingh and Malhotra, Richa and Satheesh, Aiswarya Padmaja and Khuntia, Satya Prakash and Sreekanthreddy, Peddagangannagari and Varghese, Linu and Kolla, Srivalli and Chandrani, Pratik and Choughule, Anuradha and Pange, Priyanka and others},
  journal={JCO Global Oncology},
  volume={9},
  pages={e2300047},
  year={2023},
  publisher={Wolters Kluwer Health},
  doi={10.1200/GO.23.00047}
}

@article{chalmers2017analysis,
	author = {Chalmers, Zachary R. and Connelly, Caitlin F. and Fabrizio, David and Gay, Laurie and Ali, Siraj M. and Ennis, Riley and Schrock, Alexa and Campbell, Brittany and Shlien, Adam and Chmielecki, Juliann and Huang, Franklin and He, Yuting and Sun, James and Tabori, Uri and Kennedy, Mark and Lieber, Daniel S. and Roels, Steven and White, Jared and Otto, Geoffrey A. and Ross, Jeffrey S. and Garraway, Levi and Miller, Vincent A. and Stephens, Phillip J. and Frampton, Garrett M.},
	journal = {Genome Medicine},
	number = {1},
	pages = {34},
	title = {Analysis of 100,000 human cancer genomes reveals the landscape of tumor mutational burden},
	volume = {9},
	year = {2017},
	doi = {10.1186/s13073-017-0424-2}}

@article{luo2025ancestral,
    author = {Luo, Mei and Yang, Jingwen and Schäffer, Alejandro A. and Chen, Chengxuan and Liu, Yuan and Chen, Yamei and Lin, Chunru and Diao, Lixia and Zang, Yong and Lou, Yanyan and Salman, Huda and Mills, Gordon B. and Ruppin, Eytan and Han, Leng},
    title = {Ancestral Differences in Anticancer Treatment Efficacy and Their Underlying Genomic and Molecular Alterations},
    journal = {Cancer Discovery},
    volume = {15},
    number = {3},
    pages = {511-529},
    year = {2025},
    month = {03},
    doi = {10.1158/2159-8290.CD-24-0827}
}

@article{gandara2025tumor,
  title={Tumor mutational burden and survival on immune checkpoint inhibition in> 8000 patients across 24 cancer types},
  author={Gandara, David R and Agarwal, Neeraj and Gupta, Shilpa and Klempner, Samuel J and Andrews, Miles C and Mahipal, Amit and Subbiah, Vivek and Eskander, Ramez N and Carbone, David P and Riess, Jonathan W and others},
  journal={Journal for Immunotherapy of Cancer},
  volume={13},
  number={2},
  pages={e010311},
  year={2025},
  doi={10.1136/jitc-2024-010311}
}

@article{sorich2025tumour,
  title={Tumour Mutational Burden and Immune Checkpoint Inhibitor Response in Non-small Cell Lung Cancer: A Continuous Modelling Approach},
  author={Sorich, Michael J and Manning-Bennett, Arkady T and Li, Lee X and Shahnam, Adel and Kichenadasse, Ganessan and Karapetis, Christos S and Abuhelwa, Ahmad Y and McKinnon, Ross A and Rowland, Andrew and Hopkins, Ashley M},
  journal={Targeted Oncology},
  volume={20},
  number={2},
  pages={361--369},
  year={2025},
  publisher={Springer},
  doi={10.1007/s11523-024-01124-2}
}

@incollection{rollin2005two,
  title={A Two-stage Design for Choosing Among Experimental Treatments in Clinical Trials},
  author={Rollin, Linda and Chen, Pinyuen},
  booktitle={Advances in Ranking and Selection, Multiple Comparisons, and Reliability: Methodology and Applications},
  pages={385--409},
  year={2005},
  publisher={Springer},
  doi={10.1007/0-8176-4422-9_20}
}

@article{xie2013sequential,
  title={Sequential bayes-optimal policies for multiple comparisons with a known standard},
  author={Xie, Jing and Frazier, Peter I},
  journal={Operations Research},
  volume={61},
  number={5},
  pages={1174--1189},
  year={2013},
  publisher={INFORMS},
  doi={10.1287/opre.2013.1207}
}

@article{kim2006asymptotic,
  title = {On the Asymptotic Validity of Fully Sequential Selection Procedures for Steady-State Simulation},
  author = {Kim, Seong-Hee and Nelson, Barry L.},
  journal = {Operations Research},
  volume = {54},
  number = {3},
  pages = {475--488},
  year = {2006},
  publisher = {INFORMS},
  doi = {10.1287/opre.1060.0281}
  }

@article{chick2012sequential,
author = {Chick, Stephen E. and Frazier, Peter},
title = {Sequential Sampling with Economics of Selection Procedures},
journal = {Management Science},
volume = {58},
number = {3},
pages = {550-569},
year = {2012},
doi = {10.1287/mnsc.1110.1425}
}

@article{kim2005comparison,
  title={Comparison with a standard via fully sequential procedures},
  author={Kim, Seong-Hee},
  journal={ACM Transactions on Modeling and Computer Simulation (TOMACS)},
  volume={15},
  number={2},
  pages={155--174},
  year={2005},
  publisher={ACM New York, NY, USA},
  doi={
10.1145/1060576.1060579}
}

@article{nelson2001comparisons,
  title={Comparisons with a standard in simulation experiments},
  author={Nelson, Barry L and Goldsman, David},
  journal={Management Science},
  volume={47},
  number={3},
  pages={449--463},
  year={2001},
  publisher={INFORMS},
  doi={10.1287/mnsc.47.3.449.9778}
}

@inproceedings{chen2017nearly,
  title={Nearly instance optimal sample complexity bounds for top-k arm selection},
  author={Chen, Lijie and Li, Jian and Qiao, Mingda},
  booktitle={Artificial Intelligence and Statistics},
  pages={101--110},
  year={2017},
  organization={PMLR},
  url =  {https://proceedings.mlr.press/v54/chen17a.html}
}

@inproceedings{jun2016top,
  title={Top arm identification in multi-armed bandits with batch arm pulls},
  author={Jun, Kwang-Sung and Jamieson, Kevin and Nowak, Robert and Zhu, Xiaojin},
  booktitle={Artificial Intelligence and Statistics},
  pages={139--148},
  volume =  {51},
  year={2016},
  organization={PMLR},
  url =  {https://proceedings.mlr.press/v51/jun16.html}
}

@article{shivam2026asymptotic,
author = {Shivam and Bhargab Chattopadhyay and Nil Kamal Hazra},
title = {On asymptotic optimality of sequential Gini index estimation under complex survey design with sub-stratification},
journal = {Sequential Analysis},
volume = {},
number = {},
pages = {1--44},
year = {2026},
publisher = {Taylor \& Francis},
doi = {10.1080/07474946.2026.2671036}
}

@article{even2006action,
  title={Action elimination and stopping conditions for the multi-armed bandit and reinforcement learning problems.},
  author={Even-Dar, Eyal and Mannor, Shie and Mansour, Yishay and Mahadevan, Sridhar},
  journal={Journal of machine learning research},
  volume={7},
  number={6},
  year={2006},
  Pages= {1079 --1105},
  url={https://jmlr.csail.mit.edu/papers/volume7/evendar06a/evendar06a.pdf}
}

@inproceedings{karnin2013almost,
  title={Almost optimal exploration in multi-armed bandits},
  author={Karnin, Zohar and Koren, Tomer and Somekh, Oren},
  booktitle={International conference on machine learning},
  volume =  {28},  
  number =       {3},
  pages={1238--1246},
  year={2013},
  organization={PMLR},
  url =  {https://proceedings.mlr.press/v28/karnin13.html}
}

@article{BD2005,
  title={Asymptotic inference from multi-stage samples},
  author={Bhattacharya, Debopam},
  journal={Journal of Econometrics},
  volume={126},
  number={1},
  pages={145--171},
  year={2005},
  publisher={Elsevier},
  doi={10.1016/j.jeconom.2004.01.002}
}

@article{BD2007,
  title={Inference on inequality from household survey data},
  author={Bhattacharya, Debopam},
  journal={Journal of Econometrics},
  volume={137},
  number={2},
  pages={674--707},
  year={2007},
  publisher={Elsevier},
  doi={10.1016/j.jeconom.2005.09.003}
}

@article{CASTRILLONCANDAS2022104885,
title = {Anomaly detection: A functional analysis perspective},
journal = {Journal of Multivariate Analysis},
volume = {189},
pages = {104885},
year = {2022},
issn = {0047-259X},
doi = {https://doi.org/10.1016/j.jmva.2021.104885},
author = {Julio E. Castrillón-Candás and Mark Kon}
}

@article{XING2026105640,
title = {To minimize the expected total sampling cost in sequential testing about a random vector},
journal = {Journal of Multivariate Analysis},
volume = {215},
pages = {105640},
year = {2026},
issn = {0047-259X},
doi = {10.1016/j.jmva.2026.105640},
author = {Yiming Xing}
}

@article{Dudewicz01021980,
author = {Edward J. Dudewicz},
title = {Ranking (Ordering) and Selection: An Overview of How to Select the Best},
journal = {Technometrics},
volume = {22},
number = {1},
pages = {113--119},
year = {1980},
publisher = {Taylor \& Francis},
doi = {10.1080/00401706.1980.10486108}
}

@article{BINDER,
  title={Estimating some measures of income inequality from survey data: an application of the estimating equations approach},
  author={Binder, David A and Kovacevic, Milorad S},
  journal={Survey Methodology},
  volume={21},
  pages={137--146},
  year={1995},
  publisher={Statistics Canada},
  url={https://www150.statcan.gc.ca/n1/pub/12-001-x/1995002/article/14396-eng.pdf}
}

@article{INCOME,
  title={Income distribution functions with disturbances},
  author={Ransom, Michael R and Cramer, Jan S},
  journal={European Economic Review},
  volume={22},
  number={3},
  pages={363--372},
  year={1983},
  publisher={Elsevier},
  doi={10.1016/0014-2921(83)90050-8}
}
\end{document}